\RequirePackage[l2tabu, orthodox]{nag}		
\documentclass[10pt]{article}

\usepackage[T1]{fontenc}
\usepackage{lmodern}
\usepackage{microtype}											
\microtypesetup{expansion=false}						
\usepackage{natbib}       									
\usepackage{amsmath}                        
\numberwithin{equation}{section}
\usepackage{graphicx} 											
\usepackage{enumitem}												
\usepackage{mdwlist}												
\usepackage{cancel}
\usepackage[dvipsnames]{xcolor}							
\usepackage[plainpages=false, pdfpagelabels]{hyperref} 
	\hypersetup{
		colorlinks   = true,
		citecolor    = RoyalBlue, 
		linkcolor    = RubineRed, 
		urlcolor     = Turquoise
	}
\usepackage{amssymb}                        
\usepackage[mathscr]{eucal}                 
\usepackage{dsfont}													
\usepackage[paperwidth=8.5in,paperheight=11in,top=1.25in, bottom=1.25in, left=1.0in, right=1.0in]{geometry} 
\usepackage{mathtools}                      
	\mathtoolsset{showonlyrefs=true}          
\usepackage{amsthm}                         
	\allowdisplaybreaks                       
	\theoremstyle{plain}
	\newtheorem{theorem}{Theorem}
	\numberwithin{theorem}{section}
	\newtheorem{lemma}[theorem]{Lemma}       	
	\newtheorem{proposition}[theorem]{Proposition}
	\newtheorem{corollary}[theorem]{Corollary}
	\theoremstyle{definition}
	
	\newtheorem{example}[theorem]{Example}
	\newtheorem{remark}[theorem]{Remark}
	\newtheorem{assumption}[theorem]{Assumption}

\newcommand{\<}{\langle}
\renewcommand{\>}{\rangle}
\renewcommand{\(}{\left(}
\renewcommand{\)}{\right)}
\renewcommand{\[}{\left[}
\renewcommand{\]}{\right]}

\newcommand\Cb{\mathds{C}}
\newcommand\Eb{\mathds{E}}
\newcommand\Fb{\mathds{F}}

\newcommand\Gb{\mathds{G}}
\newcommand\Hb{\mathds{H}}
\newcommand\Pb{\mathds{P}}

\newcommand\Rb{\mathds{R}}

\newcommand\Nb{\mathds{N}}

\newcommand\Fc{\mathscr{F}}
\newcommand\Gc{\mathscr{G}}
\newcommand\Hc{\mathscr{H}}

\newcommand\Nc{\mathscr{N}}
\newcommand\Oc{\mathscr{O}}

\newcommand\eps{\varepsilon}
\newcommand\om{\omega}
\newcommand\Om{\Omega}
\newcommand\sig{\sigma}

\newcommand\Lam{\Lambda}

\newcommand\Gam{\Gamma}
\newcommand\lam{\lambda}
\newcommand\del{\delta}
\newcommand\Del{\Delta}

\newcommand\varphih{\widehat{\varphi}}

\newcommand\Nt{\widetilde{N}}

\newcommand{\co}{\eta}
\newcommand{\indicator}{\mathds{1}}

\renewcommand\d{\partial}
\newcommand{\ind}{\perp \! \! \! \perp}
\newcommand\ii{\mathtt{i}}
\newcommand\dd{\mathrm{d}}
\newcommand\ee{\mathrm{e}}
\renewcommand\Re{\operatorname{Re}}
\renewcommand\Im{\operatorname{Im}}

\begin{document}

\title{Robust Replication of Volatility and Hybrid Derivatives on Jump Diffusions}

\author{
Peter Carr
\thanks{Department of Finance and Risk Engineering, NYU Tandon. \textbf{e-mail}: \url{petercarr@nyu.edu}
}
\and
Roger Lee
\thanks{
Department of Mathematics, University of Chicago.  \textbf{e-mail}: \url{rogerlee@math.uchicago.edu }
}
\and
Matthew Lorig
\thanks{Department of Applied Mathematics, University of Washington.  \textbf{e-mail}: \url{mlorig@uw.edu}}
}

\date{This version: \today}

\maketitle

\begin{abstract}
We price and replicate a variety of claims written on the log price $X$ and quadratic variation $[X]$ of a risky asset, modeled as a positive semimartingale, subject to stochastic volatility and jumps.  The pricing and hedging formulas do not depend on the dynamics of volatility process,
aside from integrability and independence assumptions; in particular, the volatility process may be non-Markovian and exhibit jumps of unknown distribution.  The jump risk may be driven by any finite activity Poisson random measure with bounded jump sizes.  As hedging instruments, we use the underlying risky asset, a zero-coupon bond, and European calls and puts with the same maturity as the claim to be hedged.  Examples of contracts that we price include variance swaps, volatility swaps, a claim that pays the realized Sharpe ratio, and a call on a leveraged exchange traded fund.
\end{abstract}


%
%

\section{Introduction}
\label{sec:intro}
Consider an underlying risky asset, which exhibits both stochastic volatility and independent jumps.  In this setting, we show how to value claims on the $\log$ price of the asset and its quadratic variation relative to vanilla European puts and calls.  
{Under an additional assumption that jump sizes of the $\log$ price of the risky asset are restricted to a discrete finite set}, we show how to replicate claims on the  $\log$ price of the asset and its quadratic variation by dynamically trading zero-coupon bonds, shares of the underlying and a portfolio of European puts and calls.
\par
The class of models we consider in this paper is semi-parametric in a sense we now describe.  The distribution and arrival rate of jumps of the risky asset must be specified parametrically.  However, we do not specify a particular volatility process.  Rather, we simply require that the volatility process be an adapted righ-continuous process that evolves independently of the Brownian motion and Poisson random measure that drive the price process of the risky asset.  In particular, the volatility process may be non-Markovian and it may experience jumps.  Because we need not specify a particular volatility process, our pricing formula and replication strategies are robust to misspecification of the volatility process.
\par
{This paper is the updated version of the working paper \cite{rrvd}, which showed how to price and replicate claims on the quadratic variation of the $\log$ price of a risky asset \emph{without} jumps.}  That work was  extended in \cite{lorig-carr-lee-2} where the authors show how to value and replicate a variety of barrier-style claims on the $\log$ price and quadratic variation of a risky asset \emph{without} jumps.  In both papers, the underlying is assumed to have continuous sample paths and an independent volatility-driving process.  These assumptions imply a symmetric model-induced implied volatility smile.  Symmetric smiles are observed in certain markets (e.g., FX), but generally are not observed for options on equity, where smiles typically exhibit downward sloping at-the-money skews.
\par
Matching the skew of implied volatility is important both for pricing and hedging, and there are a number of ways this can be achieved.
One method of matching skew is to use Dupire's formula \cite{dupire1994pricing} to find the local volatility model that is consistent the market's quoted call and put prices.  Another means of matching skew 
{is to consider a stochastic volatility model such Heston \cite{heston1993} or SABR \cite{sabr}.  In these models, the correlation between the $\log$ price and volatility processes can be adjusted in order to match the observed implied volatility skew. A third means of matching the skew is the approach taken in this paper:}
to consider models that allow the underlying risky asset to experience jumps; asymmetric jumps induce asymmetric smiles.
While both local volatility and jump models can match quoted option prices, the corresponding delta hedges differ significantly.   
The delta computed from the local volatility model typically falls below the Black-Scholes delta (as computed using a given option's implied volatility), whereas the delta computed from the model with jumps typically falls above the Black-Scholes delta.  As, empirically, delta is above the Black-Scholes delta (for options on SPX),
this is motivation for matching the skew of implied volatility with jumps rather than local volatility.
{Another reason for considering models with jumps is that, consistent with empirical observations, these models induce an explosion of the at-the-money skew as time to maturity approaches zero.  By contrast, the implied volatility skews induced by stochastic volatility models such as Heston and SABR remain bounded as time to maturity approaches zero.}
\par
The rest of this paper proceeds as follows.  In Section \ref{sec:model} we describe a market for a risky asset and state our modeling assumptions.  In Section \ref{sec:pricing} we show how to price power-exponential-style claims on $\log$ price and its quadratic variation and in Section \ref{sec:replication} we show these claims can be replicated.  Lastly, in Section \ref{sec:example}, we price a variety of claims that do not fall into the power-exponential category.

%
%

\section{Model and assumptions}
\label{sec:model}
We fix a finite time horizon $T < \infty$ and consider a frictionless market, defined on a filtered probability space $(\Om, \Fc, \Fb, \Pb)$ satisfying the usual conditions, such that the prices of all assets are martingales with respect to $(\Fb,\Pb)$.  The probability measure $\Pb$ represents the market's chosen pricing measure and the filtration $\Fb = (\Fc_t)_{0 \leq t \leq T}$ represents this history of the market.

Assume $\Fb = \Gb \vee \Hb$ where $\Gb=(\Gc_t)_{0 \leq t \leq T}$ and $\Hb=(\Hc_t)_{0 \leq t \leq T}$ are independent filtrations,
let $W$ be a $\Gb$-Brownian motion, let $\sig$ be a $\Gb$-adapted right-continuous process independent of $W$, and let $N$ be a Poisson random measure
with respect to $\Hb$, with intensity measure $\nu(\dd z)\dd u$ for some L\'evy measure $\nu$.
  
Let $B_t$ be the price of a zero-coupon bond paying one unit of currency at time $T$.  Assuming zero interest rates, or that all prices are expressed as $T$-forward prices, we have $B_t = 1$ for all $t \in [0,T]$.
Let $S_t$ be the price of a risky asset, which pays no dividends.  Suppose $S$ is strictly positive and has dynamics of the form
\begin{align}
\dd S_t
	&=	\sig_t S_t \dd W_t + \int_\Rb (\ee^z - 1) S_{t-} \Nt(\dd t, \dd z) , &
\Nt(\dd t, \dd z)
	&=	N(\dd t, \dd z) - \nu(\dd z) \dd t ,
\end{align}
where $W$ is a Brownian motion and $\Nt$ is the compensated Poisson random measure with respect to $(\Fb,\Pb)$.  
{We refer the reader to \cite[Ch. 1]{oksendal2} for an overview of L\'evy-It\^o processes.}
We will not specify dynamics for the \emph{volatility process} $\sig$.  Note that $\sig$ may be non-Markovian and may experience jumps.  However, we have required that $\sig$ evolve independently of $W$ and $N$.
For simplicity, we further assume there exist constants $b, c < \infty$ such that
\begin{align}
\int_0^T \sig_t^2 \dd t
	&< b , &
\nu(\Rb)
	&< \infty , &
\nu( |z| > c )
	&=	0 . \label{eq:assumption}
\end{align}
For certain claims, conditions \eqref{eq:assumption} can be relaxed, {as described in \cite[Section 8]{rrvd}}.  However, our aim is not to provide here the most general conditions under which our pricing and hedging methodology can be applied.  Rather, we aim to provide simple conditions, which allow us to clearly illustrate our pricing and hedging methods without complicating the presentation with numerous technicalities.
\par
The $\log$ price of the risky asset $X_t := \log S_t$ therefore has dynamics
\begin{align}
\dd X_t
	&=	- \tfrac{1}{2} \sig_t^2 \dd t + \sig_t \dd W_t - \int_\Rb (\ee^z - 1 - z) \nu(\dd z) \dd t + \int_\Rb z \Nt(\dd t, \dd z) . \label{eq:dX}
\end{align}
\par
Let $P_t(K)$ and $C_t(K)$ be the time-$t$ prices of, respectively, a European put and European call written on $S$, 
maturing at time $T$ with strike $K$.  Under the assumptions above, 
\begin{align}
P_t(K)
	&=	\Eb_t (K - S_T)^+ , &
C_t(K)
	&=	\Eb_t (S_T - K)^+ , &
t
	&\in [0,T] , &
K
	&\geq 0 ,
\end{align}
where the notation $\Eb_t \, \, \cdot \, := \Eb[ \, \cdot \, | \Fc_t ]$ denotes conditional expectation.
As $S = \ee^X$, we may refer to claims written on $S$ or $X$ interchangeably, with the understanding that these are the same thing.
Our payoff decompositions will assume a European put or call trades at every strike $K >0$.  As \cite{breeden} show, this assumption is equivalent to knowing the distribution of $S_T$ under $\Pb$.  Additionally,  \cite{carrmadan1998} show that this assumption allows general $T$-expiry European claims on $S_T$ to be perfectly replicated with a static portfolio of bonds, puts, and calls; for general function $f$ that can be expressed as the difference of convex functions, the resulting pricing formula, under integrability conditions, is
\begin{align}
\Eb_t f(S_T)
	&=		f(S_t) B_t + \int_0^{S_t}  f''(K) P_t(K) \dd K + \int_{S_t}^\infty f''(K) C_t(K) \dd K, \label{eq:f.call.put}
\end{align}
where $f'$ is the left-derivative of $f$ and $f''$ is the second derivative, which exists as a generalized function.
While in reality calls and puts trade at only finitely many strikes, this can be addressed following techniques described in \cite{leung-lorig}, who show how to optimally adjust static hedges when calls and puts are traded at only discrete strikes in a finite interval.

\begin{remark}[Limitations of our modeling framework]
\label{rmk:limitations}
{Our modeling framework has certain limitations, which we describe here.  First, because we have assumed that $\sig$ evolves independently of $W$ and $N$, the class of models we consider cannot capture {correlation between instantaneous volatility and price}.  Nevertheless, the errors that would result from using the pricing and replication  strategies developed in this paper in a setting in which $\sig$ and $W$ are correlated can, to an extent, be minimized using a \textit{correlation immunization strategy}, which is described in \cite[Section 4]{rrvd}.  Extensive Monte Carlo testing of the correlation immunization strategy have been carried out \cite{lorig-lin}.
Second, the assumptions in \eqref{eq:assumption} exclude most traditional stochastic volatility models and exponential L\'evy models because the former do not typically have bounded integrated variance and the latter do not typically have bounded jump sizes.  However, {this assumption can be relaxed, as discussed after \eqref{eq:assumption}. Moreover,} our aim is not to consider a class of models that includes all other models.  Rather, our aim is to consider a class of models that captures the dynamics of the market, and we are not aware of any empirical evidence that 
{the market is better described by a traditional SV model than by SV dynamics in which integrated variance is capped at, for instance, $10^{100}$}.
}
\end{remark}

\begin{remark}[Relation to other work]
\label{rmk:previous-work}
{The present paper {\emph{initiated}} a line of work in the general area of robust pricing and replication of claims on realized variance.
{An earlier version of this paper, the unpublished working paper \cite{rrvd}, developed pricing and replication strategies} for claims on $X$ and $[X]$ under an assumption that $X$ experiences no jumps.
In \cite{lorig-carr-lee-2}, the results of \cite{rrvd} are extended to knock-in, knock-out and rebate claims written on $X$ and $[X]$.
And in \cite{carr2011variance} and \cite{lorig-carr-lee-2}, variance swaps are robustly priced when $X$ is a time-changed L\'evy process and time-changed Markov process, respectively.}
\end{remark}

%
%

\section{Pricing power-exponential claims}
\label{sec:pricing}
Let $[X]$ denote the quadratic variation of the $X$ process.  By \eqref{eq:dX}, we have
\begin{align}
\dd [X]_t
	&=	\sig_t^2 \dd t + \int_\Rb z^2 N(\dd t, \dd z) .
\end{align}
This section will price and replicate the real and imaginary parts of a \emph{power-exponential claim}, which we define as any claim whose payoff has the form
\begin{align}
	&\text{Power-exponential Claim Payoff}:&
	&X_T^n [X]_T^m \ee^{\ii \om X_T + \ii \co [X]_T} ,
	&n,m \in \{0\} \cup \Nb , &
	&\om,\co  \in \Cb .
\end{align}
These power-exponential claims will be used as building blocks to construct more general claims.

\begin{remark}
The various processes and random variables discussed in this section and Section \ref{sec:replication} are $\Cb$-valued.  The pricing and hedging results given below should be understood to hold for the real and imaginary components.  For example, when we say ``the price of $Z$'' we mean ``the price of the real and imaginary parts of $Z$,'' and when we say ``to replicate $Z$'' we mean ``to replicate the real and imaginary parts of $Z$.''
\end{remark}

We have the decomposition
\begin{align}
X_t
	&=	X_t^c + X_t^j ,
\end{align}
where the dynamics of the continuous component $X^c$ and the jump component $X^j$ are given by
\begin{align}
\dd X_t^c
	&=	-\tfrac{1}{2} \sig_t^2 \dd t + \sig_t \dd W_t , &
\dd X_t^j
	&=	- \int_\Rb (\ee^z - 1 - z) \nu(\dd z) \dd t + \int_\Rb z \Nt(\dd t, \dd z) .
\end{align}
Likewise, the quadratic variation process $[X]$ also separates into a continuous component $[X^c]$ and an independent jump component $[X^j]$:
\begin{align}
[X]_t
	&=	[X^c]_t + [X^j]_t , &
\dd [X^c]_t
	&=	\sig_t^2 \dd t , &
\dd [X^j]_t
	&=	\int_\Rb z^2 N(\dd t, \dd z) ,
\end{align}
Proposition \ref{prp:char1} will relate the joint $\Fc_t$-conditional characteristic function of $(X_T,[X]_T)$ to the $\Fc_t$-conditional characteristic function of $X_T$.  Its proof will use the following lemma.

\begin{lemma}
\label{lem:E1=E2}
Define $u:\Cb^2 \to \Cb$ by either of the following:
\begin{align}
u(\om,\co )
	&:=	\ii \( - \tfrac{1}{2} \pm \sqrt{\tfrac{1}{4} - \om^2 - \ii \om + 2 \ii \co} \)
	=:	u_\pm(\om,\co ) . \label{eq:u}
\end{align}
Then for all $\om,\co \in \Cb$,  
\begin{align}
\Eb_t \ee^{ \ii \om (X_T^c - X_t^c) + \ii \co ([X^c]_T - [X^c]_t)}
	&=	\Eb_t \ee^{\ii u(\om,\co ) (X_T^c - X_t^c) } .  \label{eq:E1=E2}
\end{align}
\end{lemma}

\begin{proof}
See Appendix \ref{sec:E1=E2}
\end{proof}

\begin{proposition}
\label{prp:char1}
Define $\psi: \Cb^2 \to \Cb$ by
\begin{align}
\psi(\om,\co )
	&:=	\int_\Rb \Big( \ee^{\ii \om z + \ii \co z^2} - 1 - \ii \om ( \ee^z - 1 ) \Big) \nu(\dd z). \label{eq:psi}
\end{align}
Then $(X_T,[X]_T)$ has $\Fc_t$-conditional joint characteristic function
\begin{align}
\Eb_t \ee^{\ii \om X_T  + \ii \co[X]_T }
	&=	\frac{ \ee^{(T-t)\psi(\om,\co ) + \ii (\om - u(\om,\co ))X_t + \ii \co [X]_t}}{\ee^{(T-t)\psi(u(\om,\co ),0)}} \Eb_t \ee^{\ii u(\om,\co ) X_T}, \label{eq:char1}
\end{align}
where $u: \Cb^2 \to \Cb$ is defined in \eqref{eq:u}.
\end{proposition}

\begin{proof}
See Appendix \ref{sec:char1}.
\end{proof}

\begin{corollary}
\label{cor:pow-exp}
Fix $\om,\co  \in \Cb$ and $n,m \in \{0\} \cup \Nb$.
Assume $\frac{1}{4} - \ii \om + 2 \ii \co - \om^2 \neq 0$.
Then
\begin{align}
&\Eb_t X_T^n [X]_T^m \ee^{\ii \om X_T  + \ii \co [X]_T } \\
	&=	\Eb_t \sum_{j=0}^n \sum_{k=0}^m \Bigg(
				\binom{n}{j} \binom{m}{k}
				(-\ii \d_\om)^j (-\ii \d_\co)^k \frac{ \ee^{(T-t)\psi(\om,\co ) + \ii (\om - u(\om,\co ))X_t + \ii \co [X]_t}}{\ee^{(T-t)\psi(u(\om,\co ),0)}}
			\cdot (-\ii \d_\om)^{n-j} (-\ii \d_\co)^{m-k} \ee^{\ii u(\om,\co ) X_T} \Bigg). \label{eq:name}
\end{align}
where $u$ and $\psi$ are defined in \eqref{eq:u} and \eqref{eq:psi}, respectively.
\end{corollary}

\begin{proof}
See Appendix \ref{sec:pow-exp}.
\end{proof}

Corollary \ref{cor:pow-exp} relates the price of a (path-dependent) power-exponential claim to the price of a (path-independent) European claim written on $X_T$.  Specifically, 
\begin{align}
\Eb_t X_T^n [X]_T^m \ee^{\ii \om X_T + \ii \co [X]_T}
	&=	\Eb_t g(X_T;X_t,[X]_t) ,
\end{align}
where the function $g(\,\cdot\,;X_t,[X]_t)$ is given by the right-hand side of \eqref{eq:name} (keep in mind that $X_t,[X]_t \in \Fc_t$).
In turn, the price $\Eb_t g(X_T;X_t,[X]_t)$ of the European claim can be related to value of a portfolio consisting of vanilla European puts and calls and zero-coupon bonds, which are market observables, by setting $f(S_T) = g(\log S_T;X_t,[X]_t)$ in \eqref{eq:f.call.put}.

\begin{remark}
We describe equations of the form $\Eb \varphi(X_T,[X]_T) = \Eb g(X_T)$ by saying that the function $g$ \emph{prices} the claim with payoff $\varphi(X_T,[X]_T)$.  For any given $\varphi$, the function $g$ will not be unique.  For example, the right-hand-side of \eqref{eq:name} will depend on whether we choose $u=u_+$ or $u=u_-$.
\end{remark}

In order to apply Corollary \ref{cor:pow-exp} to price a power-exponential claim, we require an expression for the L\'evy exponent $\psi$, which can be
computed explicitly for a variety of L\'evy measures $\nu$, such as the following:
\begin{align}
\text{Dirac sum}:&&
\nu
	&=	\sum_j \lam_j \del_{m_j}, \label{eq:nu.dirac} \\
\text{Uniform}:&&
\nu(\dd z)
	&=	\lam \indicator_{\{ m_1 < z < m_2 \}} \dd z ,
			\label{eq:nu.uniform} \\
\text{Trunc. Exp.}:&&
\nu(\dd z)
	&=	\lam \indicator_{\{|z| < m\}} \ee^{-\alpha|z|} \dd z ,
\end{align}
where $\lam, \lam_j, \alpha,m >0$ and $m_1<m_2$.  From \eqref{eq:psi}, we compute
\begin{align}
\text{Dirac sum}:&&
\psi(\om,\co )
	&=	 \sum_j \lam_j \Big( \ee^{\ii \om m_j + \ii \co m_j^2} - 1 - \ii \om ( \ee^{m_j} - 1 ) \Big) , \label{eq:psi.dirac} \\
\text{Uniform}:&&
\psi(\om,\co )
	&=	\lam \sqrt{\frac{\ii\pi}{4\co} } \ee^{-\frac{\ii \om ^2}{4 \co}}
			\( \text{erf}\(\frac{\co (2 \co m_1+\om )}{2 (-\ii \co)^{3/2}}\)-\text{erf}\(\frac{\co (2 \co m_2 +\om )}{2 (-\ii \co)^{3/2}}\) \) \\ &&&\quad
			+ \lam \Big( (\ii \om - 1 )(m_2 - m_1) - (m_2 - m_1) \Big) ,			\label{eq:psi.uniform} \\
\text{Trunc. Exp.}:&&
\psi(\om,\co )
	&=	\lam \sqrt{ \frac{\ii\pi}{4\co} } \ee^{\frac{\ii (\alpha +\ii \om )^2}{4 \co}} \(
			\text{erf}\(\frac{\alpha + \ii \om+2 \ii \co m  }{2 \sqrt{-\ii \co}}\)-\text{erf}\(\frac{\alpha + \ii \om }{2 \sqrt{-\ii \co}}\)
			\) \\ &&& \quad
			+ \lam \sqrt{ \frac{\ii\pi}{4 \co} } \ee^{\frac{\ii (\alpha -\ii \om )^2}{4 \co}} \(
			\text{erf}\(\frac{\alpha - \ii \om+2 \ii \co m  }{2 \sqrt{-\ii \co}}\)-\text{erf}\(\frac{\alpha - \ii \om }{2 \sqrt{-\ii \co}}\)
			\) \\ &&& \quad
			+ \frac{2 \lam \(\ee^{-\alpha m}-1\)}{\alpha }
			- \frac{2 \lam \ii \omega  \ee^{-\alpha m}}{\alpha  \(\alpha^2-1\)}
			\Big( \alpha^2 - \alpha^2 \cosh m+\ee^{\alpha m}-\alpha  \sinh m-1\Big),
\end{align}
where $\text{erf}$ denotes the \emph{error function} defined by $\text{erf}(x):= (2/\sqrt{\pi})\int_0^x  \ee^{-z^2/2} \dd z$.

\begin{example}[Variance Swap]
\label{eq:VarSwap}
Consider the floating leg of a (continuously monitored) variance swap, which pays $[X]_T$ to the long side at time $T$.
For simplicity, let $X_0 = 0$.  Then setting $(n,m,\om,\co ) = (0,1,0,0)$ in \eqref{eq:name} we obtain $\Eb [X]_T = \Eb g(X_T;0,0)$
where
\begin{align}
g(x;0,0)
	&= -2 x + T \Big( - 2 \< \ee^{\Del X} \> + \< \Del X^2 \> + 2 \< \Del X \>  + 2 \< 1 \> \Big)	, &
	&\text{(for $u=u_+$)}  \label{eq:g0.vs.1} \\
g(x;0,0)
	&=	2 x \ee^x + T \ee^x \Big( - 2 \< \Del X \ee^{\Del X} \> + 2 \< \ee^{\Del X} \> + \< \Del X^2 \> - 2 \< 1 \> \Big), &
	&\text{(for $u=u_-$)}   \label{eq:g0.vs.2}
\end{align}
where $\< f(\Del X) \> := \int_\Rb f(z)\nu(\dd z)$.
In Figure \ref{fig:VarSwap} we plot $g(\log S_T;0,0)$ as a function of $S_T$ for both $u_+$ and $u_-$ and for various jump distributions and intensities.
\end{example}

\begin{remark}
The function $g$ in \eqref{eq:g0.vs.1} and \eqref{eq:g0.vs.2} depends on the time to maturity $T$.  This is in contrast to the results of \cite{carr2011variance} where, in a time-changed L\'evy setting, the authors find that the variance swap has the same value as a European-style $\log$ contract, whose payoff function has no dependence on time-to-maturity.  As empirical evidence from \cite{carr2011variance} indicates the European-style payoff function that prices the variance swap \emph{does} depend on time to maturity, this is motivation to consider the models presented in present paper rather than those considered in \cite{carr2011variance}.
\end{remark}

%
%

\section{Replicating exponential claims}
\label{sec:replication}
Define a complex-valued \emph{self-financing portfolio}
with respect to a $\Cb^J$-valued semimartingale $\Upsilon$
to be a $\Cb^J$-valued locally bounded predictable process $\Xi$ such that
\begin{align}
\dd \Pi_t
	&=	\sum_j \Xi_{t}^{(j)} \dd \Upsilon_t^{(j)} 
&
\text{where }\Pi_t
	&:=	\sum_j \Xi_t^{(j)} \Upsilon_t^{(j)}.
\label{eq:complex}
\end{align}
In particular, if $\Xi^{(j)}$ and $\Upsilon^{(j)}$ are real-valued for all $j$, then expression \eqref{eq:complex} corresponds to the usual notion of a self-financing portfolio.  The dynamics of the real and imaginary parts of $\Pi$ are given by
\begin{align}
\dd (\Re \Pi_t)
	&=	\sum_j  (\Re \Xi_{t-}^{(j)}) \dd (\Re \Upsilon_t^{(j)}) - \sum_j (\Im \Xi_{t-}^{(j)}) \dd (\Im \Upsilon_t^{(j)}) , \label{eq:real}\\
\dd (\Im \Pi_t)
	&=	\sum_j (\Re \Xi_{t-}^{(j)}) \dd (\Im \Upsilon_t^{(j)}) + \sum_j (\Im \Xi_{t-}^{(j)}) \dd (\Re \Upsilon_t^{(j)}) . \label{eq:imag}
\end{align}
respectively.  Thus, expression \eqref{eq:complex} should be seen as a concise way to state both \eqref{eq:real} and \eqref{eq:imag}.

\begin{assumption}
Throughout Section \ref{sec:replication}, the constants $\om,\co  \in \Cb$ are \emph{fixed} and $u \equiv u(\om,\co )$ is given by \eqref{eq:u}.
\end{assumption}

At any time $t \leq T$, by \eqref{eq:char1}, the claim on the exponential payoff $\ee^{\ii \om X_T + \ii \co [X]_T}$ has value
\begin{align}
\Eb_t \ee^{\ii \om X_T + \ii \co [X]_T}
	&=	A_t Q_t^{(u)} ,
\end{align}
where we have defined
\begin{align}
A_t
	&:=	\ee^{\ii (\om-u) X_t + \ii \co [X]_t} \frac{ \ee^{(T-t)\psi(\om,\co )}}{\ee^{(T-t)\psi(u,0)}} , &
Q_t^{(q)}
	&:=	\Eb_t \ee^{\ii q X_T} , &
q
	&\in \Cb . \label{eq:AQ.def}
\end{align}
Theorem \ref{thm:hedge} will show that $A Q^{(u)}$ is the value process of a self-financing portfolio, which
 gives a trading strategy to replicate the exponential claim because 
 $A_T Q_T^{(u)} = \ee^{\ii \om X_T + \ii \co [X]_T}$.
\par
Theorem \ref{thm:hedge} uses Lemmas \ref{lem:Delta.Q} and \ref{lem:sym}, presented below,
and the standard notation 
\begin{align}
\Del H_t
	&:=	H_t - H_{t-} = H_t - \lim_{s \nearrow t} H_s ,
\end{align}
for the jump in $H$ at time $t$, where $H$ is any process with left limits.

\begin{lemma}
\label{lem:Delta.Q}
For any $q \in \Cb$, let $Y^{(q)}$ and $Z^{(q)}$ be (the c\`adl\`ag versions of) the martingales
\begin{align}
Y_t^{(q)}
	&:=	\Eb_t \ee^{\ii q X_T^c} , &
Z_t^{(q)}
	&:=	\Eb_t \ee^{\ii q X_T^j} , & 0\leq t\leq T\label{eq:YZ}
\end{align}
Then, under the assumptions of Section \ref{sec:model}, we have
\begin{align}
\Del A_t \Del Q_t^{(q)}
	&=	\Del A_t ( Y_{t-}^{(q)}  \Del Z_t^{(q)} ) , \label{eq:DeltaQ}\\
Y_{t-}^{(q)}  \Del Z_t^{(q)}
	&=	Q_{t-}^{(q)} \int_\Rb \Big( \ee^{\ii q z} - 1 \Big) N(\dd t , \dd z) . \label{eq:DeltaZ}
\end{align}
\end{lemma}

\begin{proof}
See Appendix \ref{sec:Delta.Q}.
\end{proof}

\begin{lemma}
\label{lem:sym}
For any $q \in \Cb$ and $t\in[0,T]$, define 
\begin{align}
R_t^{(q)}
	&=	\ee^{-\ii q X_t + (T-t)\psi(-\ii - q,0)} , \label{eq:R}
\end{align}
with $\psi$ given in \eqref{eq:psi}.
Then
\begin{align}
R_t^{(q)} Q_t^{(q)}
	&=	R_t^{(-\ii-q)} Q_t^{(-\ii-q)} , \label{eq:RQ=RQ}
\end{align}
with $Q^{(q)}$ given in \eqref{eq:AQ.def}.
\end{lemma}

\begin{proof}
See Appendix \ref{sec:sym}.
%
\end{proof}

\begin{theorem}
\label{thm:hedge}
Let $q \in \Cb$.  Define processes $\Del \Gam^{(u)} = (\Del \Gam_t^{(u)})_{0 \leq t \leq T}$ and $\Del \Om^{(q)} = (\Del \Om_t^{(q)})_{0 \leq t \leq T}$ by
\begin{align}
\Del \Gamma_t^{(u)}		
	&:=	Q_{t-}^{(u)} \Del A_t - \ii (\om-u) \frac{A_{t-}Q_{t-}^{(u)}}{S_{t-}} \Del S_t
			 + \Del A_t ( Y_{t-}^{(u)} \Del Z_t^{(u)} )  \\
	&=	A_{t-} Q_{t-}^{(u)} \int_\Rb \Big(
			\ee^{\ii \om z + \ii \co z^2} - \ee^{\ii u z} - \ii (\om - u)(\ee^z - 1)
			\Big) N(\dd t, \dd z) , \label{eq:DeltaGamma} \\
\Del \Om_t^{(q)}		
	&:=	Q_{t-}^{(q)} \Del R_t^{(q)} + \ii q \frac{R_{t-}^{(q)} Q_{t-}^{(q)}}{S_{t-}} \Del S_t
			+ \Del R_t^{(q)} ( Y_{t-}^{(q)} \Del Z_t^{(q)} ) \\
	&=	R_{t-}^{(q)} Q_{t-}^{(q)} \int_\Rb \Big(
			- \ee^{\ii q z}  + 1 + \ii q (\ee^z - 1)
			\Big) N(\dd t, \dd z) . \label{eq:DeltaOmega}
\end{align}
Let $(q_1, q_2, \ldots, q_m) \in \Cb^m$.
Suppose there exists an $m$-dimensional predictable process $H = (H_t)_{0 \leq t \leq T}$ with components $H^{(j)} = (H_t^{(j)})_{0 \leq t \leq T}$ satisfying
\begin{align}
0
	&=	\Del \Gam_t^{(u)} + \sum_{j=1}^m H_t^{(j)} \Big( \Del \Om_t^{(q_j)} - \Del \Om_t^{(-\ii-q_j)} \Big) , \label{eq:H}
\end{align}
Then 
\begin{align}
\dd (A_t Q_t^{(u)})
	&=	A_{t-} \dd Q_t^{(u)} + \ii (\om-u) \frac{A_{t-}Q_{t-}^{(u)}}{S_{t-}} \dd S_t \\ &\quad
			+ \sum_{j=1}^m H_{t-}^{(j)} \Big(
			R_{t-}^{(q_j)} \dd Q_t^{(q_j)} - R_{t-}^{(-\ii-q_j)} \dd Q_t^{(-\ii-q_j)}
			+ (1 - 2 \ii q_j) \frac{R_{t-}^{(q_j)} Q_{t-}^{(q_j)}}{S_{t-}} \dd S_t
			\Big) , \label{eq:dAQ}
\end{align}
where the processes $A$, $Q^{(q)}$ and $R^{(q)}$ are as given in \eqref{eq:AQ.def} and \eqref{eq:R}.
\end{theorem}

\begin{proof}
See Appendix \ref{sec:hedge}.
\end{proof}

\begin{remark}\label{rem:hedge}
By \eqref{eq:dAQ}, the following self-financing portfolio replicates the exponential claim $\ee^{\ii \om X_T + \ii \co [X]_T}$: at all times $t < T$ one should
\begin{align}
&\bullet \, \text{hold $A_{t-}$ European claims with payoff $\ee^{\ii u X_T}$,} \\
&\bullet \, \text{hold $\Big( \ii (\om-u) \frac{A_{t-}Q_{t-}^{(u)}}{S_{t-}}+\sum_{j=1}^m  H_{t-}^{(j)} (1 - 2 \ii q_j) \frac{R_{t-}^{(q_j)} Q_{t-}^{(q_j)}}{S_{t-}} \Big)$ shares of $S$,} \\
&\bullet \, \text{for $j=1,2,\ldots,m$, hold $H_{t-}^{(j)} R_{t-}^{(q_j)}$ European claims with payoff $\ee^{\ii  q_j X_T}$,}\\
&\bullet \, \text{for $j=1,2,\ldots,m$, hold $-H_{t-}^{(j)} R_{t-}^{(-\ii-q_j)}$ European claims with payoff $\ee^{\ii (-\ii-q_j) X_T}$,}\\
&\bullet \, \text{lend and borrow zero coupon bonds $B$ from the bank as needed.}
\end{align}
This portfolio's net position in European claims, which has value
\begin{align}
A_{t-} \Eb_t \ee^{\ii u X_T} + \sum_{j=1}^m H_{t-}^{(j)}
\Big( R_{t-}^{(q_j)} \Eb_t \ee^{\ii  q_j X_T} - R_{t-}^{(-\ii-q_j)} \Eb_t \ee^{\ii (-\ii-q_j) X_T}  \Big) ,
\end{align}
can be constructed from a portfolio of European calls $C_t(K)$ and puts $P_t(K)$ for $K \geq 0$, using \eqref{eq:f.call.put}.
\end{remark}

\begin{remark}\label{rem:spanning}
{
The intuition of condition \eqref{eq:H} is that the \emph{pricing} relation 
$\Eb_t \ee^{\ii \om X_T + \ii \co [X]_T}
	=	A_t \Eb_t \ee^{\ii q X_T}$
is valid in the presence of jump risk; however, the naive candidate for a \emph{hedging} portfolio,
namely $A_{t-}$ contracts on $\ee^{\ii q X_T}$, is \emph{not} a valid replication of $\ee^{\ii \om X_T + \ii \co [X]_T}$, 
because this naive portfolio fails to self-finance at jump times.  So we augment the naive portfolio with ``zero-cost collars'', 
specifically $H^{(j)}$ units of the ``collar'' that combines the claims on payouts $\ee^{\ii  q_j X_T}$ and $\ee^{\ii (-\ii-q_j) X_T}$.
At jump times these collars have a combined profit/loss which provides the needed financing to offset the ``tracking error'' $\Del \Gam_t^{(u)}$
of the naive hedge, if \eqref{eq:H} holds.   This leads us to ask, whether there exist $H^{(j)}$ satisfying \eqref{eq:H} -- in other words, 
do the collars span the tracking error?   The answer will involve (naturally, in this \emph{spanning} context) a \emph{full rank} 
condition \eqref{eq:L} on the collars. }

{To be specific:} in order to hedge
an exponential claim with payoff $\ee^{\ii \om X_T+ \ii \co [X]_T}$,
what remains is to find a predictable process $H = (H_t)_{0 \leq t \leq T}$ with components $H^{(j)} = (H_t^{(j)})_{0 \leq t \leq T}$ satisfying \eqref{eq:H}.  This is the subject of the next proposition.
\end{remark}

\begin{proposition}
\label{prp:matrix}
Suppose the L\'evy measure $\nu$ has the form
\begin{align}
\nu
	&=	\sum_{i=1}^n \lam_i \del_{z_i}, \label{eq:nu}
\end{align}
for some $(\lam_1, \lam_2 , \ldots, \lam_n) \in \Rb_+^n$ and some $(z_1, z_2, \ldots, z_n) \in \Rb^n$.
Define an $n \times 1$ stochastic column matrix $K_t = (K_t^{(i)})$ with entries
\begin{align}
K_t^{(i)}	
	&=	A_{t-} Q_{t-}^{(u)} F(z_i) , &
F(z)
	&:=	\ee^{\ii \om z + \ii \co z^2} - \ee^{\ii u z} - \ii (\om - u)(\ee^z - 1) . \label{eq:K}
\end{align}
Suppose there exists $(q_1, q_2, \ldots , q_m) \in \Cb^m$ with $m \geq n$ such that the $n \times m$ stochastic matrix $L_t = (L_t^{(i,j)})$, with entries
\begin{align}
L_t^{(i,j)}
	&=	R_{t-}^{(q_j)} Q_{t-}^{(q_j)}  G(z_i;q_j) , &
G(z;q)
	&:=	-  \ee^{\ii q z} + \ee^{(1-\ii q)z} - (1 - 2 \ii q) (\ee^z - 1) , \label{eq:L}
\end{align}
has rank $n$ for all $t \in [0,T)$.  Then there exists an $m$-dimensional predictable process
$H$ with components $H^{(j)}$
that satisfies \eqref{eq:H}; it solves
\begin{align}
K_t
	&= L_t H_t .	\label{eq:K=LH}
\end{align}
In particular, if $m=n$ then $H_t = L_t^{-1} K_t$.
\end{proposition}

\begin{proof}
See Appendix \ref{sec:matrix}.
\end{proof}

\begin{corollary}\label{cor:twojumps}
Suppose $\nu = \lam_1 \del_{z_1} + \lam_2 \del_{z_2}$, {where $z_1 z_2(z_1-z_2)\neq 0$.}
{Then the exponential claim paying $\ee^{\ii \om X_T + \ii \co [X]_T}$ is replicated by the hedging strategy of Remark \ref{rem:hedge}}, with 
\begin{align}\label{eq:twojumps}
\[ \begin{array}{c}
H_t^{(1)} \\ H_t^{(2)}
\end{array} \]
	&=	\[ \begin{array}{cc}
			R_{t-}^{(q_1)} Q_{t-}^{(q_1)}  G(z_1;q_1)	&	R_{t-}^{(q_2)} Q_{t-}^{(q_2)}  G(z_1;q_2) \\
			R_{t-}^{(q_1)} Q_{t-}^{(q_1)}  G(z_2;q_1)	&	R_{t-}^{(q_2)} Q_{t-}^{(q_2)}  G(z_2;q_2)
			\end{array} \]^{-1}
			\[ \begin{array}{c}
			A_{t-} Q_{t-}^{(u)} F(z_1) \\
			A_{t-} Q_{t-}^{(u)} F(z_2)
			\end{array} \] .
\end{align}
{
where, given $(z_1,z_2)$, the $(q_1,q_2)$ are chosen such that for all $t$ the inverse exists.  The existence of such $(q_1,q_2)$ is a \emph{conclusion} of this Corollary, not an assumption.  }
\end{corollary}

\begin{proof}
See Appendix \ref{sec:twojumps}.
\end{proof}

\section{Pricing other payoffs}
\label{sec:example}
This section applies the results of Section \ref{sec:pricing} to price some contracts with payoffs $\varphi(X_T,[X]_T)$ that are not of the power-exponential form.  Generally speaking, our results shall take the form
\begin{align}
\Eb \varphi(X_T,[X]_T)
	&= \Eb g(X_T) , &
X_0
	&=	0 , \label{eq:g0}
\end{align}	
where $\Eb g(X_T)$ can be computed relative to traded European calls/puts via \eqref{eq:f.call.put}.  Note, by the spatial homogenity of the $X$ process, there is no loss in generality in assuming $X_0 = 0$.

\subsection{Fractional powers and ratios}

\begin{proposition}[Fractional powers of quadratic variation]
\label{prp:frac}
Consider a fractional power claim, whose payoff function is of the form
\begin{align}
\varphi(x,v)
	&=	v^r , &
0
	&< r < 1 . \label{eq:phi.frac}
\end{align}
Then
\begin{align}
g(x)
	&:=	 \frac{r}{\Gam(1-r)} \int_0^\infty \frac{1}{z^{r+1}}
			\Big( \ee^{\ii u(0,0) x} - \frac{ \ee^{T\psi(0,\ii z)}}{\ee^{T \psi(u(0,\ii z),0)}} \ee^{\ii u(0,\ii z) x} \Big) \dd z . \label{eq:g0.frac}
\end{align}
satisfies \eqref{eq:g0} and hence prices the fractional power claim.
\end{proposition}

\begin{proof}
See Appendix \ref{sec:frac}.
\end{proof}

\begin{example}[Volatility Swap]
\label{eq:VolSwap}
Consider the floating leg of a (continuously monitored) volatility swap, which pays $\sqrt{[X]_T}$ to the long side at time $T$.
The payoff function $\varphi(x,v) = \sqrt{v}$ can be obtained as a special case of \eqref{eq:phi.frac} by setting $r=1/2$.
In Figure \ref{fig:VolSwap} we plot $g(\log S_T)$ as a function of $S_T$ for various jump distributions and intensities, where $g$ is given by \eqref{eq:g0.frac}.
\end{example}

\begin{proposition}[Ratio claims (I)]
\label{prp:ratio}
Consider a ratio claim, whose payoff function has the form
\begin{align}
\varphi(x,v)
	&=	\frac{x \ee^{\ii p x} }{ (v + \eps)^{r} } , & \text{where $p\in\Cb$, \ 
$r\in (0,1)$,\ and 
$\eps>0$. }\label{eq:phi.ratio}
\end{align}
Then 
\begin{align}
g(x)
	&:=	\frac{1}{r\Gam(r)}  \int_0^\infty 
			(-\ii \d_p) \frac{ \ee^{T\psi(p,\ii z^{1/r})}}{\ee^{T \psi(u(p,\ii z^{1/r}),0)}} \ee^{\ii u(p,\ii z^{1/r}) x - z^{1/r} \eps } 
			\ \dd z.
			\label{eq:g0.ratio}
\end{align}
satisfies \eqref{eq:g0} and hence prices the ratio claim.
\end{proposition}

\begin{proof}
See Appendix \ref{sec:ratio}.
\end{proof}

\begin{example}[Realized Sharpe ratio]
\label{ex:sharpe}
The \emph{Sharpe ratio} was introduced in \cite{sharpe1966mutual} as a simple way to measure the performance of an investment while adjusting for its risk.
Define the \emph{realized Sharpe ratio}
\begin{align}
\Lam_T
	&:=		\frac{ X_T - X_0 }{ \sqrt{ [X]_T-[X]_0 } }.
\end{align}
Consider a claim that pays the realized Sharpe ratio.  With $X_0 = [X]_0 = 0$ we have $\varphi(X_T,[X]_T) =  X_T / \sqrt{[X]_T}$.  The payoff function $\varphi(x,v) = x/\sqrt{v}$ can be approximated with arbitrary accuracy by setting $r=1/2$ in Proposition \ref{prp:ratio} and choosing $\eps$ small enough.  Figure \ref{fig:SharpeRatio} plots $g(\log S_T)$ as a function of $S_T$, where $g$ is given by \eqref{eq:g0.ratio}.
\end{example}

\begin{proposition}[Ratio claims (II)]
\label{prp:negative}
Consider a ratio claim whose payoff function has the form
\begin{align}
\varphi(x,v)
	&=	\frac{ \ee^{\ii p x} }{ (v + \eps)^r } , &
r, \eps
	&> 0,\ p\in\Cb. \label{eq:phi.negative}
\end{align}
Then
\begin{align}
g(x)
	&:=	\frac{1}{r\Gam(r)}  \int_0^\infty 
			\frac{ \ee^{T\psi(p,\ii z^{1/r})}}{\ee^{T \psi(u(p,\ii z^{1/r}),0)}} \ee^{\ii u(p,\ii z^{1/r}) x - z^{1/r} \eps } \dd z
			\label{eq:g0.negative}
\end{align}
satisfies \eqref{eq:g0} and hence prices the ratio claim.
\end{proposition}

\begin{proof}
See Appensix \ref{sec:negative}.
\end{proof}

\begin{remark}
\label{rmk:hedging}
{Throughout this section we have used the fact that a large class of payoffs of the form $\varphi(X_T,[X]_T)$ can be written as derivatives, sums and/or integrals of exponential basis functions $\ee^{\ii \om X_T + \ii \eta [X]_T}$.  By linearity, one can in principle {combine} the replication strategies developed in Section \ref{sec:replication} in order to replicate {a further expanded class of} payoffs of the form $\varphi(X_T,[X]_T)$.}
\end{remark}


\subsection{Options on Levered Exchange Traded Funds}
\label{sec:LETF}
A growing class of \emph{exchange-traded funds} (ETFs) are the \emph{leveraged exchanged traded funds} (LETFs).
In an ideal setting (i.e., no management fees), the relationship between an LETF $L = (L_t)_{0 \leq t \leq T}$ and the underlying ETF $S=(S_t)_{0 \leq t \leq T}$ is
\begin{align}
\frac{\dd L_t}{L_{t-}}
	&=	\beta \frac{\dd S_t}{S_{t-}} ,
\end{align}
where $\beta$ is a fixed constant known as the \emph{leverage ratio}.  Typical values of $\beta$ are $\{-3,-2,-1,2,3\}$.  As \cite{avellaneda1} point out, the value of $L_T$ depends on the entire path of $S$ over the interval $[0,T]$.  This is most readily seen by looking at $(Y_t)_{0 \leq t \leq T}$, the $\log$ LETF process: $Y_t = \log L_t$.  With the dynamics of $X = \log S$ given by \eqref{eq:dX}, a simple application of the It\^o formula yields
\begin{align}
\dd Y_t
	&=	\dd Y_t^c + \dd Y_t^j , \label{eq:dY} \\
\dd Y_t^c
	&=	\beta \dd X_t^c + \tfrac{1}{2}\beta(1-\beta) \dd [X^c]_t , \label{eq:dYc} \\ 
\dd Y_t^j
	&=	- \int_\Rb \Big( \beta( \ee^z - 1 ) - \log \big( \beta( \ee^z - 1 ) + 1 \big) \Big) \nu(\dd z) \dd t 
			+ \int_\Rb \log \big( \beta( \ee^z - 1 ) + 1 \big)  \Nt(\dd t, \dd z) ,  \label{eq:dYj}
\end{align}
where we assume that the constant $c$ appearing in \eqref{eq:assumption} satisfies
\begin{align}
\beta( \ee^z - 1 ) + 1
	&> 0 , &
	&\forall \, z \in [-c,c] , \label{eq:assumption2}
\end{align}
which guarantees that, when the ETF $S$ jumps, the LETF $L$ jumps to a strictly positive value.
Observe that $\dd Y_t$ depends not only on $\dd X_t^c$ but also on $\dd [X^c]_t$ and on a nontrivial integral with respect to the Poisson random measure $\Nt(\dd t, \dd z)$.  Because of the intricate path-dependent behavior, there has been significant interest in relating option prices/implied volatilities on $X$ to option prices/implied volatilities written on $Y$; see, for example,
\cite{haugh-letf,leungsircarLETF,lorig-leung-pascucci-1,lee2015leverage}.
Although $Y_T$ cannot be written as a function of $X_T$ and $[X]_T$ only, 
our framework allows us to value a claim written on $Y_T$ (which can be viewed as a claim on the path of $X$) relative to a European (i.e., path-independent) claim written on $X_T$.
\par
The following proposition relates the characteristic function of $(Y_T - Y_t)$, conditional on $\Fc_t$, to the characteristic function of $(X_T-X_t)$, also conditional on $\Fc_t$.

\begin{proposition}
\label{prp:LETF}
Let $X$ and $Y$ have dynamics given by \eqref{eq:dX} and \eqref{eq:dY}, respectively.  Define $\chi : \Cb \to \Cb$ by
\begin{align}
\chi(q)
	&:=	\int_\Rb \Big( \big( \beta(\ee^z -1) + 1 \big)^{\ii q} - 1 - \ii q \beta(\ee^z-1) \Big) \nu(\dd z).  \label{eq:chi}
\end{align}
Then the characteristic function of $(Y_T - Y_t)$, conditional on $\Fc_t$, is given by
\begin{align}
\Eb_t \ee^{\ii q (Y_T - Y_t)}
	&=	\frac{\ee^{(T-t)\chi(q)}}{\ee^{(T-t)\psi(u( q \beta , q \tfrac{1}{2} \beta(1-\beta) ),0)}}
			\Eb_t \ee^{ \ii u( q \beta , q \tfrac{1}{2} \beta(1-\beta) )(X_T - X_t) } , \label{eq:char.XY}
\end{align}
where $u$ and $\psi$ are given by \eqref{eq:u} and \eqref{eq:psi}, respectively.
\end{proposition}

\begin{proof}
See Appendix \ref{sec:LETF2}.
\end{proof}

Using Proposition \ref{prp:LETF}, we can relate the value of a claim written on $Y$ to the value of a European claim written on $X$.

\begin{theorem}
\label{thm:LETF}
Let $\varphih$ be the generalized (one-dimensional) Fourier transform of $\varphi : \Rb \to \Rb$.  We have
\begin{align}
\varphih(q)
	&=	\frac{1}{2\pi} \int_\Rb \ee^{-\ii q x} \varphi(x)\dd x , &
q
	&\in \Cb .
\end{align}
Define $q_r := \Re q$ and $q_i := \Im q$.  
Assume the inverse Fourier transform of $\varphih$ is $\varphi$
\begin{align}
\varphi(x)
	&=	\int_\Rb \ee^{\ii q x} \varphih(q) \dd q_r. \label{eq:iFT}
\end{align}
Assume further that $\varphih(\cdot + \ii q_i)$ has no singularities and satisfies
\begin{align}
|\varphih(q)|
	&= \Oc( |q_r|^{-1-\eps} ) &
	&\text{as $|q_r| \to \infty$} , \label{eq:phih.assumption}
\end{align}
for some $\eps > 0$.
Then, with $X$ and $Y$ given by \eqref{eq:dX} and \eqref{eq:dY}, respectively, we have
\begin{align}
\Eb_t \varphi(Y_T)
	&=	\Eb_t g(X_T;X_t,Y_t) , \\
g(x;X_t,Y_t)
	&=	\int_\Rb \varphih(q) 
			\frac{\ee^{\ii q Y_t + (T-t)\chi(q)}}{ \ee^{(T-t)\psi(u( q \beta , q \tfrac{1}{2} \beta(1-\beta) ),0)} } 
			\ee^{ \ii u( q \beta , q \tfrac{1}{2} \beta(1-\beta) )(x - X_t) }\dd q_r, \label{eq:g.LETF}
\end{align}
where $u$, $\psi$ and $\chi$ are given by \eqref{eq:u}, \eqref{eq:psi} and \eqref{eq:chi}, respectively.
\end{theorem}

\begin{proof}
See Appendix \ref{sec:LETF3}.
\end{proof}

\begin{example}[LETF Call option]
Consider a call option written on the LETF.  The payoff function $\varphi(y) := (\ee^y - \ee^k)^+$ has a generalized Fourier transform
\begin{align}
\varphih(q)
	&=	\frac{-\ee^{k-\ii k q}}{2 \pi (q^2+ \ii q)} , &
q_i := \Im q
	&< - 1 .
\end{align}
Observe that $| \varphih(q) | = \Oc(|q_r|^{-2})$ as $|q_r| \to \infty$, where $q_r := \Re q$.  Moreover, with $q_i < - 1$ fixed, the function $\varphih(\cdot + \ii q_i):\Rb \to \Cb$ has no singularities.  Thus, $\varphih$ satisfies the conditions of Theorem \ref{thm:LETF}.  In Figure \ref{fig:LETF} we plot the function $g(\log S_T;X_t,Y_t)$ with $g$ given by \eqref{eq:g.LETF} as a function of $S_T$ for various leverage ratios $\beta$ and for both $u=u_+$ and $u=u_-$.
\end{example}

%
%

\section{Conclusion}
\label{sec:conclude}
In this paper we consider a variety of claims written on the $\log$ price $X$ and quadratic variation $[X]$ of a risky asset $S = \ee^X$.  The asset $S$ is modeled as a positive semimartingale with finite activity jumps and independent unspecified (possibly non-Markovian) volatility.
In this setting, we show how to price various path-dependent claims relative to path-independent calls and puts on $S$.  We also show how some of these path-dependent claims can be replicated by trading the underlying $S$, a bond $B$, and calls and puts on $S$.  A number of examples are provided in which we explicitly compute a payoff function $g$ of a European claim whose value equals the value of the path-dependent claim.

%
%

\appendix
 
\section{Appendix}

\subsection{Proof of Lemma \ref{lem:E1=E2}}
\label{sec:E1=E2}
Recall the characteristic function of a normal random variable
\begin{align}
\Eb \ee^{\ii \om Z}
	&=	\ee^{\ii m \om - \tfrac{1}{2} a^2 \om^2} , &
Z
	&\sim \Nc( m , a^2 ) . \label{eq:char.Z}
\end{align}
Let $\Fc_T^\sig$ denote the sigma-algebra generated by $(\sig_t)_{0 \leq t \leq T}$.  Then $([X^c]_T - [X^c]_t) \in \Fc_T^\sig$ and
\begin{align}
X_T^c - X_t^c | \Fc_T^\sig
	&\sim \Nc( m , a^2 ) , &
m
	&=	-\tfrac{1}{2} ( [X^c]_T - [X^c]_t ) , &
a^2
	&=	[X^c]_T - [X^c]_t . \label{eq:X.normal}
\end{align}
Using \eqref{eq:char.Z} and \eqref{eq:X.normal},
\begin{align}
\Eb_t \ee^{ \ii \om (X_T^c - X_t^c) + \ii \co([X^c]_T - [X^c]_t)}
	&=	\Eb_t \ee^{\ii \co([X^c]_T - [X^c]_t)} \Eb_t[ \ee^{ \ii \om (X_T^c - X_t^c)}| \Fc_T^\sig ]  \\
	&=	\Eb_t \Eb_t[ \ee^{ (\ii\co - ( \om^2 + \ii \om)/2 )([X^c]_T - [X^c]_t)} | \Fc_T^\sig ] &
	&		\text{(by \eqref{eq:char.Z} and \eqref{eq:X.normal})}\\
	&=	\Eb_t \Eb_t[ \ee^{ ( - ( u^2(\om,\co ) + \ii u(\om,\co ))/2 )([X^c]_T - [X^c]_t)} | \Fc_T^\sig ] &
	&		\text{(by \eqref{eq:u})}\\
	&=	\Eb_t \Eb_t[ \ee^{ \ii u(\om,\co ) (X_T^c - X_t^c)}| \Fc_T^\sig ]  &
	&		\text{(by \eqref{eq:char.Z} and \eqref{eq:X.normal})} \\
	&=	\Eb_t \ee^{ \ii u(\om,\co ) (X_T^c - X_t^c)}  ,
\end{align}
which establishes \eqref{eq:E1=E2}.

\subsection{Proof of Proposition \ref{prp:char1}}
\label{sec:char1}
Because $X^c$ is $\Gb$-adapted and $X^j$ is $\Hb$-adapted and $\Fb=\Gb\vee\Hb$ where $\Gb$ and $\Hb$ are independent,
\begin{align}
\Eb_t \ee^{\ii u(\om,\co ) (X_T - X_t)}
	&=	\Eb_t \ee^{\ii u(\om,\co ) (X_T^c - X_t^c)} \Eb_t \ee^{\ii u(\om,\co ) (X_T^j - X_t^j)} \\
	&=	\Eb_t \ee^{\ii \om (X_T^c - X_t^c) + \ii \co ([X^c]_T - [X^c]_t)} \Eb_t \ee^{\ii u(\om,\co ) (X_T^j - X_t^j)} , \label{eq:char2}
\end{align}
where the second equality uses \eqref{eq:E1=E2}.  Similarly,
\begin{align}
\Eb_t \ee^{\ii \om (X_T - X_t) + \ii \co ([X]_T - [X]_t)}
	&=	\Eb_t \ee^{\ii \om (X_T^c - X_t^c) + \ii \co ([X^c]_T - [X^c]_t)}
			\Eb_t \ee^{\ii \om (X_T^j - X_t^j) + \ii \co ([X^j]_T - [X^j]_t)} \label{eq:char3}\\
	&=	\Eb_t \ee^{\ii u(\om,\co ) (X_T - X_t)}
			\frac{\Eb_t \ee^{\ii \om (X_T^j - X_t^j) + \ii \co ([X^j]_T - [X^j]_t)}}{\Eb_t \ee^{\ii u(\om,\co ) (X_T^j - X_t^j)}} . \label{eq:char4}\\
	&=	\Eb_t \ee^{\ii u(\om,\co ) (X_T - X_t)}
			\frac{\ee^{(T-t)\psi(\om,\co )}}{\ee^{(T-t)\psi(u(\om,\co ),0)}}. \label{eq:char5}
\end{align}
by \eqref{eq:char2} and applying the L\'evy-Khintchine formula to the two-dimensional L\'evy process $(X^j,[X^j])$,
whose characteristic exponent $\psi$, given by \eqref{eq:psi}, is well-defined for all $(\om,\co )\in \Cb^2$ due to \eqref{eq:assumption}.  Rearranging \eqref{eq:char5} produces \eqref{eq:char1}.

\subsection{Proof of Corollary \ref{cor:pow-exp}}
\label{sec:pow-exp}
By Proposition \ref{prp:char1},
\begin{align}
\Eb_t X_T^n [X]_T^m \ee^{\ii \om X_T  + \ii \co [X]_T }
	&=	(-\ii \d_\om)^n (-\ii \d_\co)^m \Eb_t \ee^{\ii \om X_T  + \ii \co [X]_T } \\
	&=	(-\ii \d_\om)^n (-\ii \d_\co)^m \frac{ \ee^{(T-t)\psi(\om,\co ) + \ii (\om - u(\om,\co ))X_t + \ii \co [X]_t}}{\ee^{(T-t)\psi(u(\om,\co ),0)}} \Eb_t \ee^{\ii u(\om,\co ) X_T} \\
	&=	\text{R.H.S. of \eqref{eq:name}} ,
\end{align}
where the interchanges of differentiation and expectation in the first and last equalities are justified since, for any $n,m \in \{ 0 \} \cup \Nb $ and $\om,\co  \in \Cb$, there exists a constant $c_1 > 0$ such that
\begin{align}
| \d_\om^n \d_\co^m \ee^{\ii \om x + \ii \co v} |
	&< c_1 \ee^{ c_1 (|x| + |v|)}, &
\Eb_0 c_1 \ee^{ c_1 (|X_T| + [X]_T)}
	&< \infty ,
\end{align}
where the finiteness of the expectation follows from \eqref{eq:assumption}.

\subsection{Proof of Lemma \ref{lem:Delta.Q}}
\label{sec:Delta.Q}
By independence of $\Gb$ and $\Hb$, 
\begin{equation*}
Q_t^{(q)}=Y_t^{(q)}Z_t^{(q)}.
\end{equation*}
By iterated expectations and the countability of $\mathscr{J}:=\{t: \Delta Y^{(q)}_t\neq 0\}$,
\begin{equation}
\Pb(\text{$N$ and $Y^{(q)}$ have a common jump time})
	= \Eb \sum_{t\in\mathscr{J}}\Pb( \Delta N_t\neq 0 \ |\ Y^{(q)})
	=  0  ,
\end{equation}
where the last step is because $N$ is still Poisson given $Y^{(q)}$, by independence.
Moreover, all jump times of $A$ are jump times of $N$, hence jump times of $A$ are \emph{not} jump times of $Y^{(q)}$,
and \eqref{eq:DeltaQ} follows.
Next, we have
\begin{align}
\Del Z_t^{(q)}
	&=	  Z_t^{(q)} - Z_{t-}^{(q)}  \\
	&=	  \ee^{\ii q X_t^j} \Eb_t \ee^{\ii q (X_T^j - X_t^j)} - \ee^{\ii q X_{t-}^j} \Eb_{t-} \ee^{\ii q (X_T^j - X_{t-}^j)}     \\
	&=	  \ee^{\ii q (X_{t-}^j+\Del X_t^j)} \Eb_t \ee^{\ii q (X_T^j - X_t^j)} - \ee^{\ii q X_{t-}} \Eb_{t-} \ee^{\ii q (X_T^j - X_{t-}^j)}      \\
	&=	  \ee^{\ii q (X_{t-}^j+\Del X_t^j)} \Eb_{t-} \ee^{\ii q (X_T^j - X_{t-}^j)} - \ee^{\ii q X_{t-}^j} \Eb_{t-} \ee^{\ii q (X_T^j - X_{t-}^j)}     \\
	&=	  \ee^{\ii q \Del X_t^j} \Eb_{t-} \ee^{\ii q X_T^j } - \Eb_{t-} \ee^{\ii q X_T^j }     \\
	&=	  \Big( \ee^{\ii q \Del X_t^j} - 1 \Big)  \Eb_{t-} \ee^{\ii q X_T^j }     \\
	&=		Z_{t-}^{(q)} \int_\Rb \Big( \ee^{\ii q z} - 1 \Big) N(\dd t , \dd z).
\end{align}
Multiplying by $Y_{t-}^{(q)}$ produces
\eqref{eq:DeltaZ}.

\subsection{Proof of Lemma \ref{lem:sym}}
\label{sec:sym}
By independence of $\Gb$ and $\Hb$, we have
\begin{align}
\Eb_t \ee^{\ii q (X_T - X_t)}
	&=	\Eb_t \ee^{\ii q (X_T^c - X_t^c)} \Eb_t \ee^{\ii q (X_T^j - X_t^j)} \\
	&=	\Eb_t \Eb_t[ \ee^{\ii q (X_T^c - X_t^c)} | \Fc_T^\sig] \Eb_t \ee^{\ii q (X_T^j - X_t^j)} \\
	&=	\Eb_t \ee^{\tfrac{1}{2}(-q^2-\ii q)([X^c]_T - [X^c]_t)} \Eb_t \ee^{\ii q (X_T^j - X_t^j)} , \label{eq:1}
\end{align}
where the third equality uses \eqref{eq:char.Z} and \eqref{eq:X.normal}.
Next, noting that with $h(q):=q^2+\ii q$ we have $h(q)=h(-\ii - q)$, it follows from \eqref{eq:1} that
\begin{align}
\frac{\Eb_t \ee^{\ii q (X_T - X_t)}}{ \Eb_t \ee^{\ii q (X_T^j - X_t^j)}}
	&=	\frac{\Eb_t \ee^{\ii (-\ii - q) (X_T - X_t)}}{ \Eb_t \ee^{\ii (-\ii - q) (X_T^j - X_t^j)}}  \label{eq:2}
\end{align}
(unless either denominator is zero, but in that case, \eqref{eq:RQ=RQ} holds because $Q_t^{(q)}=Q_t^{(-\ii-q)} =0$).
Expression \eqref{eq:RQ=RQ} follows from \eqref{eq:char5} and \eqref{eq:2}.

\subsection{Proof of Theorem \ref{thm:hedge}}
\label{sec:hedge}
We compute
\begin{align}
\dd A_t
	&=	( \ldots ) \dd t + \ii (\om-u) \frac{A_{t-}}{S_{t-}} \dd S_t + \Del A_t - \ii (\om-u) \frac{A_{t-}}{S_{t-}} \Del S_t , \label{eq:dA} \\
\Del A_t
	&=	A_{t-} \int_\Rb \Big(  \ee^{\ii (\om-u) z + \ii \co z^2} - 1 \Big) N(\dd t, \dd z) , \label{eq:DeltaA} \\
\Del S_t
	&=	S_{t-} \int_\Rb  (\ee^z - 1) N(\dd t, \dd z) , \label{eq:DeltaS}
\end{align}
where, as we shall see, the $( \ldots ) \dd t$ terms will play no role.
Next,
we have from Lemma \ref{lem:Delta.Q} that
\begin{align}
\dd [ A, Q^{(u)} ]_t
	&=	( \ldots ) \dd t + \Del A_t \Del Q_t^{(u)}
	=		( \ldots ) \dd t + \Del A_t (Y_{t-}^{(u)} \Del Z_t^{(u)})  . \label{eq:AQ-covar}
\end{align}
Now, using \eqref{eq:DeltaQ}, \eqref{eq:dA}, \eqref{eq:DeltaA}, \eqref{eq:DeltaS} and \eqref{eq:AQ-covar}, we have 
\begin{align}
\dd (A_t Q_t^{(u)})
	&=	A_{t-} \dd Q_t^{(u)} + Q_{t-}^{(u)} \dd A_t + \dd [ A, Q^{(u)} ]_t \\
	&=	( \ldots ) \dd t + A_{t-} \dd Q_t^{(u)}
			+ \ii (\om-u) \frac{A_{t-}Q_{t-}^{(u)}}{S_{t-}} \dd S_t + \Del \Gamma_t^{(u)}, \label{eq:dAQ-proof}
\end{align}
where $\Del \Gam_t^{(u)}$ is defined in \eqref{eq:DeltaGamma}.
Likewise, 
\begin{align}
\dd R_t^{(q)}
	&=	( \ldots ) \dd t - \ii q \frac{R_{t-}^{(q)}}{S_{t-}} \dd S_t + \Del R_t^{(q)} + \ii q \frac{R_{t-}^{(q)}}{S_{t-}} \Del S_t , \label{eq:dR} \\
\Del R_t^{(q)}
	&=	R_{t-}^{(q)} \int_\Rb \Big(  \ee^{-\ii q z } - 1 \Big) N(\dd t, \dd z) ,
\end{align}
from which
\begin{align}
\dd (R_t^{(q)} Q_t^{(q)})
	&=	( \ldots ) \dd t + R_{t-}^{(q)} \dd Q_t^{(q)}
			- \ii q \frac{R_{t-}^{(q)} Q_{t-}^{(q)}}{S_{t-}} \dd S_t + \Del \Om_t^{(q)}, \label{eq:dRQ}
\end{align}
where $\Del \Om_t^{(q)}$ is defined in \eqref{eq:DeltaOmega}.
Note that we have used $\Del R_t^{(q)}\Del Q_t^{(q)} =	\Del R_t^{(q)} (Y_{t-}^{(q)} \Del Z_t^{(q)})$, which follows by replacing $A$ with $R^{(q)}$ in Lemma \ref{lem:Delta.Q} and its proof.
Next, from \eqref{eq:RQ=RQ} and \eqref{eq:dRQ} we have
\begin{align}
0
	&=	\dd (R_t^{(q)} Q_t^{(q)}) - \dd (R_t^{(-\ii-q)} Q_t^{(-\ii-q)}) \\
	&=	( \ldots ) \dd t
			+ R_{t-}^{(q)} \dd Q_t^{(q)} - R_{t-}^{(-\ii-q)} \dd Q_t^{(-\ii-q)}
			+ (1 - 2 \ii q) \frac{R_{t-}^{(q)} Q_{t-}^{(q)}}{S_{t-}} \dd S_t
			+ \Del \Om_t^{(q)} - \Del \Om_t^{(-\ii-q)} . \label{eq:dRQ-dRQ}
\end{align}
Finally, combining \eqref{eq:H}, \eqref{eq:dAQ-proof} and \eqref{eq:dRQ-dRQ},
\begin{align}
\dd (A_t Q_t^{(u)})
	&=	A_{t-} \dd Q_t^{(u)} + \ii (\om-u) \frac{A_{t-}Q_{t-}^{(u)}}{S_{t-}} \dd S_t \\ &\quad
			+ \sum_{j = 1}^{m} H_{t-}^{(j)} \Big(
			R_{t-}^{(q_j)} \dd Q_t^{(q_j)} - R_{t-}^{(-\ii-q_j)} \dd Q_t^{(-\ii-q_j)}
			+ (1 - 2 \ii q_j) \frac{R_{t-}^{(q_j)} Q_{t-}^{(q_j)}}{S_{t-}} \dd S_t
			\Big) ,
\end{align}
where the $( \ldots ) \dd t$ terms must vanish since the processes $A Q^{(u)}$, $S$ and $Q^{(q)}$ are martingales.

\subsection{Proof of Proposition \ref{prp:matrix}}
\label{sec:matrix}
From \eqref{eq:DeltaGamma} and \eqref{eq:DeltaOmega} we observe that
\begin{align}
\Del \Gam_t		
	&=	A_{t-} Q_{t-}^{(u)} \int_\Rb F(z) N(\dd t, \dd z) , \\
\Del \Om_t^{(q)} - \Del \Om_t^{(-\ii-q)}
	&=	R_{t-}^{(q)} Q_{t-}^{(q)} \int_\Rb G(z;q) N(\dd t, \dd z) .
\end{align}
From \eqref{eq:nu}, we see that $N(\dd t, \Rb) \in \{ 0 \} \cup \{ z_1, z_2, \ldots z_n \}$.  Thus, in order for \eqref{eq:H} to hold, we must have
\begin{align}
A_{t-} Q_{t-}^{(u)} F(z_i)
	&=	\sum_{j=1}^m  H_t^{(j)} R_{t-}^{(q_j)} Q_{t-}^{(q_j)} G(z_i;q_j) , &
i
	&=	1, 2, \ldots , n . \label{eq:system}
\end{align}
From \eqref{eq:K} and \eqref{eq:L}, we see that \eqref{eq:system} is given in matrix notation by \eqref{eq:K=LH}.

\subsection{Proof of Corollary \ref{cor:twojumps}}\label{sec:twojumps}

{
Let $\ii\mathbb{R}\subset\mathbb{C}$ denote the imaginary axis.}  
{By Proposition \ref{prp:matrix}, given $z_1, z_2$, 
we need only verify the existence of $q_1,q_2$.  It suffices to choose $q_1\in\ii\mathbb{R}\setminus\{0,-\ii/2,-\ii\}$ arbitrarily, 
and to choose $q_2\in\ii\mathbb{R}$ such that $D(q_2)\neq 0$ where $D(q):=G(z_1;q_1)G(z_2;q)-G(z_1;q)G(z_2;q_1)$; the existence of such $q_2$ 
is clear because $|D(q)|\to\infty$ as $q\to\pm\ii\infty$.  Moreover, for $q_1,q_2\in \ii\mathbb{R}$, the $RQ$ factors in \eqref{eq:twojumps} never vanish, 
hence the invertibility condition holds.  }

\subsection{Proof of Proposition \ref{prp:frac}}
\label{sec:frac}
We have from \cite[equation (1.2.3)]{schurger2002laplace} that
\begin{align}
v^r
	&=	\frac{r}{\Gam(1-r)} \int_0^\infty \frac{1}{z^{r+1}} \Big( 1-\ee^{-z v} \Big)\dd z , &
0
	&<r<1 . \label{eq:frac}
\end{align}
Thus
\begin{align}
\Eb [X]_T^r
	&=	\frac{r}{\Gam(1-r)} \int_0^\infty \frac{1}{z^{r+1}} \Eb \Big( 1-\ee^{-z [X]_T} \Big)\dd z  &
			&\text{(by \eqref{eq:frac} and Tonelli)} \\
	&=	\frac{r}{\Gam(1-r)} \int_0^\infty \frac{1}{z^{r+1}} \Eb
			\Big( \ee^{\ii u(0,0) X_T} - \frac{ \ee^{T\psi(0,\ii z)}}{\ee^{T \psi(u(0,\ii z),0)}} \ee^{\ii u(0,\ii z) X_T} \Big)\dd z  &
			&\text{(by \eqref{eq:char1})} \\
	&=	\frac{r}{\Gam(1-r)} \Eb \int_0^\infty \frac{1}{z^{r+1}}
			\Big( \ee^{\ii u(0,0) X_T} - \frac{ \ee^{T\psi(0,\ii z)}}{\ee^{T \psi(u(0,\ii z),0)}} \ee^{\ii u(0,\ii z) X_T} \Big) \dd z &
			&\text{(by Fubini)} \\
	&=	\Eb g(X_T) , &
			&\text{(by \eqref{eq:g0.frac})}
\end{align}
where the use of Fubini is justified as follows.  Define
\begin{align}
Z(\co)
	&:=	T \Big( \psi(0,\ii \co) - \psi(u(0,\ii \co),0) \Big) + \ii u(0,\ii \co) X_T . \label{eq:Z.def}
\end{align}
Consider the case $u = u_+$; the case $u=u_-$ is analogous.  Using $u_+(0,0) = 0$ and \eqref{eq:Z.def},
\begin{align}
\Eb \Big| \ee^{\ii u(0,0) X_T} - \frac{ \ee^{T\psi(0,\ii \co)}}{\ee^{T \psi(u(0,\ii \co),0)}} \ee^{\ii u(0,\ii \co) X_T} \Big|
	&=	\Eb \Big| 1 - \ee^{Z(\co)} \Big| . \label{eq:E.stuff}
\end{align}
Observe that
\begin{align}
\Big( \Eb \big| 1 - \ee^{Z(\co)} \big| \Big)^2
	&\leq \Eb \Big| 1 - \ee^{Z(\co)} \Big|^2
	=		\Eb \Big( 1 + \ee^{2 \Re {Z(\co)}} - \ee^{ \Re {Z(\co)}} 2 \cos \Im {Z(\co)} \Big) . \label{eq:square}
\end{align}
By \eqref{eq:u} and \eqref{eq:psi},
\begin{align}
\ii u(0,\ii \co)
	&=	\tfrac{1}{2} - \sqrt{\tfrac{1}{4} - 2  \co} , \\
\psi(0,\ii \co) - \psi(u(0,\ii \co),0)
	&=	\int_\Rb \Big( \ee^{- \co z^2} - \ee^{ \ii u(0,\ii \co) z} - \ii u(0,\ii \co)(\ee^z - 1) \Big) \nu(\dd z).
\end{align}
Noting that $0 \leq \Re( \ii u(0,\ii \co) ) \leq 1$ and recalling from \eqref{eq:assumption} that $\nu(\Rb) < \infty$ and $\nu(|z|>c)=0$, we have
\begin{align}
\sup_{\co \in \Rb_+} \Re \Big( \psi(0,\ii \co) - \psi(u(0,\ii \co),0) \Big)
	&\leq \int_\Rb \Big( 1 + \ee^{ c } + |\ee^c - 1| \Big)\nu(\dd z)
	=			\nu(\Rb) \Big( 1 + \ee^{ c } + |\ee^c - 1| \Big), \\
\inf_{\co \in \Rb_+} \Re \Big( \psi(0,\ii \co) - \psi(u(0,\ii \co),0) \Big)
	&\geq \int_\Rb \Big( - \ee^{ c } - |\ee^c - 1| \Big)\nu(\dd z)
	=			\nu(\Rb) \Big( - \ee^{ c } - |\ee^c - 1| \Big).
\end{align}
Thus, from \eqref{eq:Z.def}, we conclude that $\Re Z(\co)$ is bounded uniformly in $\co$.  Combining the uniform bound of $\Re Z(\co)$ with
\eqref{eq:E.stuff} and \eqref{eq:square}, it follows that
\begin{align}
\Eb \Big| \ee^{\ii u(0,0) X_T} - \frac{ \ee^{T\psi(0,\ii \co)}}{\ee^{T \psi(u(0,\ii \co),0)}} \ee^{\ii u(0,\ii \co) X_T} \Big|
	&= \Oc(1) , &
	&\text{as $\co \to \infty$} . \label{eq:s2infty}
\end{align}
On the other hand, for $\co$ small enough, we have $\ii u(0,\ii \co) \in \Rb$, hence
\begin{align}
\Big( \Eb |1-\ee^{Z(\co)}| \Big)^2
	&\leq \Eb \Big|1-\ee^{Z(\co)} \Big|^2
	=	\Eb \Big(1 + \ee^{2 Z(\co)} - 2 \ee^{ Z(\co)} \Big) . &
	&\text{(for $\co$ small enough)} \label{eq:e1}
\end{align}
Next, observe that
\begin{align}
\Eb \ee^{Z(\co)}
	&=	\ee^{T(\psi(0,\ii \co)-\psi(u(0,\ii \co),0))} \Eb \ee^{ \ii u(0,\ii \co) X_T} \\
	&=	\ee^{T(\psi(0,\ii \co)-\psi(u(0,\ii \co),0))} \Eb \ee^{ \ii u(0,\ii \co) X_T^j} \Eb \ee^{ \ii u(0,\ii \co) X_T^c} \\
	&=	\ee^{T(\psi(0,\ii \co)-\psi(u(0,\ii \co),0))} \ee^{ T \psi(u(0,\ii \co),0) } \Eb \ee^{ - \co [X^c]_T} \\
	&=	\ee^{T \psi(0,\ii \co) } \Eb \ee^{ - \co [X^c]_T} \\
	&=	1 - \Big( M'(0) + \< \Del X^2 \> \Big) \co + \Oc(\co^2) , &
	&\text{as $\co \to 0$} , \label{eq:e2}
\end{align}
where $M(t) := \Eb \ee^{ t [X^c]_T}$ and $\< f(\Del X) \> := \int_\Rb f(z) \nu(\dd z)$.
Here, we are using that $M$ is an entire function, which follows from \eqref{eq:assumption} and \cite[Lemma 25.6]{sato1999levy}.
We also have
\begin{align}
\Eb \ee^{2 Z(\co)}
	&=	\ee^{2T(\psi(0,\ii \co)-\psi(u(0,\ii \co),0))} \Eb \ee^{ 2 \ii u(0,\ii \co) X_T} \\
	&=	\ee^{2T(\psi(0,\ii \co)-\psi(u(0,\ii \co),0))} \Eb \ee^{ 2 \ii u(0,\ii \co) X_T^j} \Eb \ee^{ 2 \ii u(0,\ii \co) X_T^c} \\
	&=	\ee^{2T(\psi(0,\ii \co)-\psi(u(0,\ii \co),0))} \ee^{ T \psi(2u(0,\ii \co),0) } \Eb \ee^{ - w(\co) [X^c]_T} \\
	&=	1 - 2 \Big( M'(0)+\< \Del X^2 \> \Big) \co + \Oc(\co^2) &
	&\text{as $\co \to 0$}, \label{eq:e3}
\end{align}
where $w(\co)=\frac{1}{2} ( 8 \co +\sqrt{1-8 \co}-1 )$ (for $u=u_+$) solves $\ii u(0,\ii w(\co)) = 2 \ii u(0,\ii \co)$ so that
\begin{align}
\Eb \ee^{-w(\co) [X^c]_T }
	&=	\Eb \ee^{\ii u(0,\ii w(\co)) X_T^c}
	=		\Eb \ee^{2 \ii u(0,\ii w(\co)) X_T^c} .
\end{align}
Inserting \eqref{eq:e2} and \eqref{eq:e3} into \eqref{eq:e1}, we obtain $(\Eb | 1-\ee^{Z(\co)} | )^2
	=	\Oc(\co^2) $ hence
\begin{align}
\Eb \big| 1-\ee^{Z(\co)} \big|
	&=	\Oc(\co) , &
	&\text{as $\co \to 0$} . \label{eq:s2zero}
\end{align}
By \eqref{eq:E.stuff}, \eqref{eq:s2infty}, and \eqref{eq:s2zero},
\begin{align}
\int_0^\infty \, \frac{1}{\co^{r+1}} \Eb
			\Big| \ee^{\ii u(0,0) X_T} - \frac{ \ee^{T\psi(0,\ii \co)}}{\ee^{T \psi(u(0,\ii \co),0)}} \ee^{\ii u(0,\ii \co) X_T} \Big|\dd \co
			& < \infty ,
\end{align}
justifying the use of Fubini.

\subsection{Proof of Proposition \ref{prp:ratio}}
\label{sec:ratio}
We have from \cite[equation (1.0.1)]{schurger2002laplace} that
\begin{align}
\frac{ x \ee^{\ii p x} }{ ( v + \eps )^{r} }
	&=	\frac{1}{r\Gam(r)} \int_0^\infty  x \ee^{\ii p x - z^{1/r} (v + \eps)}\ \dd z, &
r
	&>	0 . \label{eq:ratio}
\end{align}
hence
\begin{align}
\Eb \frac{X_T \ee^{\ii p X_T} }{ ( [X]_T + \eps )^r }
	&=	\frac{1}{r\Gam(r)} \Eb  \int_0^\infty X_T \ee^{\ii p X_T - z^{1/r} ([X]_T + \eps)} \ \dd z&
			&\text{(by \eqref{eq:ratio})} \\
	&=	\frac{1}{r\Gam(r)} \int_0^\infty \Eb X_T \ee^{\ii p X_T - z^{1/r} ([X]_T + \eps)} \ \dd z&
			&\text{(by Fubini)} \\
	&=	\frac{1}{r\Gam(r)} \int_0^\infty (-\ii \d_p) \Eb \ee^{\ii p X_T - z^{1/r} [X]_T - z^{1/r} \eps} \ \dd z&
			&\text{(by Leibniz)} \\
	&=	\frac{1}{r\Gam(r)} \int_0^\infty (-\ii \d_p)
			\frac{ \ee^{T\psi(p,\ii z^{1/r})}}{\ee^{T \psi(u(p,\ii z^{1/r}),0)}} \Eb \ee^{\ii u(p,\ii z^{1/r}) X_T - z^{1/r} \eps} \ \dd z&
			&\text{(by \eqref{eq:char1})} \\
	&=	\frac{1}{r\Gam(r)} \Eb \int_0^\infty (-\ii \d_p)
			\frac{ \ee^{T\psi(p,\ii z^{1/r})}}{\ee^{T \psi(u(p,\ii z^{1/r}),0)}} \ee^{\ii u(p,\ii z^{1/r}) X_T - z^{1/r} \eps} \ \dd z&
			&\text{(by Fubini)} \\
	&=	\Eb g(X_T) . &
			&\text{(by \eqref{eq:g0.ratio})}
\end{align}
The use of the Leibniz has already been justified in the proof of Corollary \ref{cor:pow-exp}.  The first use of Fubini's Theorem is justified since $\Eb | X_T \ee^{\ii p X_T - z^{1/r} [X]_T} | \leq \Eb | X_T \ee^{\ii p X_T } |  < \infty$, for all $p \in \Cb$, and $z \geq 0$, which implies
\begin{align}
\int_0^\infty \, \Eb \Big| X_T \ee^{\ii p X_T - z^{1/r} [X]_T  } \Big| \ee^{- z^{1/r} \eps} \dd z 
	&< \infty .
\end{align}
The second application of Fubini is justified as follows.  Define
\begin{align}
Y(p,\co)
	&:=	T \Big( \psi(p,\ii \co) - \psi(u(p,\ii \co),0) \Big) + \ii u(p,\ii \co) X_T . \label{eq:Y.def}
\end{align}
Observe that
\begin{align}
(-\ii \d_p) \frac{ \ee^{T\psi(p,\ii z^{1/r})}}{\ee^{T \psi(u(p,\ii z^{1/r}),0)}} \ee^{\ii u(p,\ii z^{1/r}) X_T }
	&=	- \ii \ee^{Y(p,z^{1/r})} \d_p Y(p,z^{1/r}) , \\
\Big| (-\ii \d_p) \frac{ \ee^{T\psi(p,\ii z^{1/r})}}{\ee^{T \psi(u(p,\ii z^{1/r}),0)}} \ee^{\ii u(p,\ii z^{1/r}) X_T } \Big|
	&= \ee^{\Re Y(p,z^{1/r})} \big| \d_p Y(p,z^{1/r}) \big| . \label{eq:abs1}
\end{align}
From \eqref{eq:u} and \eqref{eq:psi} we have
\begin{align}
\ii u_\pm(p,\ii \co)
	&=	\tfrac{1}{2} \mp \sqrt{\tfrac{1}{4} -p^2 - \ii p - 2  \co} , \\
\psi(p,\ii \co) - \psi(u(p,\ii \co),0)
	&=	\int_\Rb \Big( \ee^{\ii p z - \co z^2} - \ee^{ \ii u(p,\ii \co) z} - ( \ii p - \ii u(p,\ii \co) \big) (\ee^z - 1) \Big) \nu(\dd z).
\end{align}
Noting that, for any $a,b \in \Rb$ we have
\begin{align}
\Big| \Re \sqrt{a + \ii b} \Big|
	&= \sqrt{ \frac{ (a^2+b^2)^{1/2} + a }{2} }, &
	&\text{which implies}&
\sup_{a \leq \overline{a} } \Big| \Re \sqrt{a + \ii b} \Big|
	&=	\sqrt{ \frac{ (\overline{a}^2+b^2)^{1/2} + \overline{a} }{2} } ,
\end{align}
it follows that there exists a constant $c_1$ such that
\begin{align}
\sup_{\co \in \Rb_+} \Big| \Re \ii u(p,\ii \co) \Big|
	&< c_1 , &
\sup_{\co \in \Rb_+} \Big| \Re \Big( \psi(p,\ii \co) - \psi(u(p,\ii \co),0) \Big) \Big|
	&< c_1 , \label{eq:bounds}
\end{align}
where the second inequality follows from \eqref{eq:assumption}, the uniform bound on $| \Re \ii u(p,\ii \co) |$ and
\begin{align}
\Re \Big( \psi(p,\ii \co) - \psi(u(p,\ii \co),0) \Big)
	&=	\int_\Rb \Big( \ee^{- p_i z - \co z^2} \cos ( p_r z )- \ee^{ \Re \ii u(p,\ii \co) z} \cos (\Im \ii u(p,\ii \co) z) \Big)\ \nu(\dd z) 
			\\ &\quad
			- ( - p_i - \Re \ii u(p,\ii \co) \big) \int_\Rb (\ee^z - 1)\ \nu(\dd z) .
\end{align}
Now, observe that
\begin{align}
\d_p Y(p,\co)
	&=	T \int_\Rb \Big( \ii z \ee^{\ii p z - \co z^2} - \ee^{ \ii u(p,\ii \co) z} \d_p  \ii u(p,\ii \co) z - \ii (\ee^z - 1) \Big)\ \nu(\dd z)  
			+ \ii \d_p u(p,\ii \co) X_T , \\
\d_p u(p,\ii \co)
	&=	\frac{1-2 \ii p}{\sqrt{-4 p^2-4 \ii p-8 \co+1}} , \label{eq:dpu}
\end{align}
from which 
\begin{align}
|\d_p Y(p,\co) |
	&\leq	T \int_\Rb\Big( |z| \ee^{- p_r z - \co z^2} + \ee^{ \Re \ii u(p,\ii \co) z} | \d_p  \ii u(p,\ii \co) z|  + |\ee^z - 1 | \Big)
\ \nu(\dd z) 
				+ | \d_p u(p,\ii \co) X_T | . \label{eq:dpY}
\end{align}
Combining \eqref{eq:Y.def}, \eqref{eq:abs1}, \eqref{eq:bounds}, \eqref{eq:dpu} and \eqref{eq:dpY},
\begin{align}
\Eb \Big| (-\ii \d_p) \frac{ \ee^{T\psi(p,\ii z^{1/r})}}{\ee^{T \psi(u(p,\ii z^{1/r}),0)}} \ee^{\ii u(p,\ii z^{1/r}) X_T } \Big|
	&=	\Eb \ee^{\Re Y(p,z^{1/r})} \big| \d_p Y(p,z^{1/r}) \big|
	 =	\Oc(1) , &
	&\text{as $z \to \infty$} . \label{eq:z2infty}
\end{align}
Next, for any $a \in \Cb$, $r \in (0,1)$ and $\eps > 0$, 
\begin{align}
\int_0^\infty \Big| \frac{\ee^{-\eps z^{1/r}}}{\sqrt{a-z^{1/r}}} \Big| \dd z 
	&< \infty . \label{eq:finite}
\end{align}
By \eqref{eq:z2infty} and \eqref{eq:finite},
\begin{align}
\int_0^\infty  
\Eb \Big| (-\ii \d_p) \frac{ \ee^{T\psi(p,\ii z^{1/r})}}{\ee^{T \psi(u(p,\ii z^{1/r}),0)}} \ee^{\ii u(p,\ii z^{1/r}) X_T } \Big|
\ee^{- \eps z^{1/r}} \dd z
	&<	\infty ,
\end{align}
justifying the use of Fubini.

\subsection{Proof of Proposition \ref{prp:negative}}
\label{sec:negative}
The proof is completely analogous to the proof of Proposition \ref{prp:ratio}.  The only significant change in the proof is that, since the operator $\d_p$ does not appear in \eqref{eq:g0.negative}, one no longer needs to be concerned about the singularity that appears in the expression \eqref{eq:dpu} of $\d_p u(p,\ii \co)$.  As a result, expression \eqref{eq:g0.negative} holds for all $r>0$.

\subsection{Proof of Proposition \ref{prp:LETF}}
\label{sec:LETF2}
First, we observe that $Y^j$, given by \eqref{eq:dYj}, is a L\'evy process with characteristic exponent $\chi$.  We have
\begin{align}
\Eb_t \ee^{\ii q (Y_T^j - Y_t^j )}
	&=	\ee^{(T-t)\chi(q)} . \label{eq:chi2}
\end{align}
Next, we compute
\begin{align}
\Eb_t \ee^{\ii q (Y_T - Y_t)}
	&=	\Eb_t \ee^{\ii q (Y_T^c - Y_t^c)} \Eb_t \ee^{\ii q (Y_T^j - Y_t^j)} &
	&\text{(as $Y^c \ind Y^j$)} \\
	&=	\Eb_t \ee^{\ii q \beta (X_T^c - X_t^c) + \ii q \tfrac{1}{2} \beta(1-\beta) ([X^c]_T - [X^c]_t)} \Eb_t \ee^{\ii q (Y_T^j - Y_t^j)} &
	&\text{(by \eqref{eq:dYc})} \\
	&=	\Eb_t \ee^{ \ii u( q \beta , q \tfrac{1}{2} \beta(1-\beta) )(X_T^c - X_t^c) } \Eb_t \ee^{\ii q (Y_T^j - Y_t^j)} &
	&\text{(by \eqref{eq:E1=E2})} \\
	&=	\Eb_t \ee^{ \ii u( q \beta , q \tfrac{1}{2} \beta(1-\beta) )(X_T - X_t) } 
			\frac{ \Eb_t \ee^{\ii q (Y_T^j - Y_t^j)} }{ \Eb_t \ee^{ \ii u( q \beta , q \tfrac{1}{2} \beta(1-\beta) )(X_T^j - X_t^j) } } &
	&\text{(as $X^c \ind X^j$)} \\
	&=	\frac{\ee^{(T-t)\chi(q)}}{\ee^{(T-t)\psi(u( q \beta , q \tfrac{1}{2} \beta(1-\beta) ),0)}}
			\Eb_t \ee^{ \ii u( q \beta , q \tfrac{1}{2} \beta(1-\beta) )(X_T - X_t) } . &
	&\text{(by \eqref{eq:char5} and \eqref{eq:chi2})} 
\end{align}
Thus, we have established \eqref{eq:char.XY}.

\subsection{Proof of Theorem \ref{thm:LETF}}
\label{sec:LETF3}
We compute
\begin{align}
\Eb_t \varphi(Y_T)
	&=	\Eb_t \int_\Rb \varphih(q) \ee^{\ii q Y_T}\dd q_r  &
	&		\text{(by \eqref{eq:iFT})} \\
	&=	\int_\Rb \varphih(q) \ee^{\ii q Y_t} \Eb_t \ee^{\ii q (Y_T-Y_t)}\dd q_r  &
	&		\text{(by Parseval)} \\
	&=	\int_\Rb \varphih(q) \frac{\ee^{\ii q Y_t + (T-t)\chi(q)}}{ \ee^{(T-t)\psi(u( q \beta , q \tfrac{1}{2} \beta(1-\beta) ),0)} } 
			\Eb_t \ee^{ \ii u( q \beta , q \tfrac{1}{2} \beta(1-\beta) )(X_T - X_t) }\dd q_r  &
	&		\text{(by \eqref{eq:char.XY})} \\
	&=	\Eb_t \int_\Rb \varphih(q) 
			\frac{\ee^{\ii q Y_t + (T-t)\chi(q)}}{ \ee^{(T-t)\psi(u( q \beta , q \tfrac{1}{2} \beta(1-\beta) ),0)} } 
			\ee^{ \ii u( q \beta , q \tfrac{1}{2} \beta(1-\beta) )(X_T - X_t) }\dd q_r  &
	&		\text{(by Fubini)} \\
	&=	\Eb_t g(X_T;X_t,Y_t) . &
	&		\text{(by \eqref{eq:g.LETF})} 
\end{align}
Parseval-style identity is allowed by \cite[Theorem 39]{titchmarsh1948introduction}.  The use of Fubini's Theorem is justified as follows.
Without loss of generality, we may assume $t=0$ and take $X_0 = Y_0 = 0$.  We must show
\begin{align}
&\int_\Rb 
			\Big| \varphih(q)  \ee^{T(\chi(q)-\psi(u( q \beta , q \tfrac{1}{2} \beta(1-\beta) ),0))}
			\ee^{ \ii u( q \beta , q \tfrac{1}{2} \beta(1-\beta) )X_T } \Big| \dd q_r\\
&= \int_\Rb 
			\Big| \varphih(q) \Big| \ee^{T(\Re \chi(q)- \Re \psi(u( q \beta , q \tfrac{1}{2} \beta(1-\beta) ),0))}
			\ee^{ \Re \ii u( q \beta , q \tfrac{1}{2} \beta(1-\beta) )X_T } \dd q_r
<  \infty . \label{eq:bound-this}
\end{align}
From \eqref{eq:u} we have
\begin{align}
\Re \ii u_\pm(q \beta, q \tfrac{1}{2} \beta(1-\beta) )
	&=	\tfrac{1}{2} \pm \sqrt{ a + \ii b } , &
a
	&=	\tfrac{1}{4}+\beta^2 (q_i^2+q_i-q_r^2) , &
b
	&=	-\beta^2 (2 q_i q_r+q_r) .
\end{align}
Noting that, for any $a,b \in \Rb$ we have 
$\Big| \Re \sqrt{a + \ii b} \Big| = \sqrt{ ((a^2+b^2)^{1/2} + a )/2}$,
it follows that there exists a constant $c_1<\infty$ such that
\begin{align}
\sup_{q_r \in \Rb} \Big| \Re \ii u_\pm(q \beta, q \tfrac{1}{2} \beta(1-\beta) ) \Big|
	&< c_1 . \label{eq:u.bound2}
\end{align}
Next, we note from \eqref{eq:psi} and \eqref{eq:chi} that
\begin{align}
\Re \psi(u,0)
	&=	\int_\Rb \Big( \ee^{\Re \ii u z} \cos (\Im \ii u z) - 1 - \Re \ii u (\ee^z - 1) \Big) \nu(\dd z), \label{eq:Re.psi} \\
\Re \chi(q)
	&=	\int_\Rb \Big( 
			\ee^{- q_i \log ( \beta(\ee^z -1) + 1 ) } \cos \big(q_r \log ( \beta(\ee^z -1) + 1 ) \big)  - 1 + q_i \beta(\ee^z-1) 
			\Big) \nu(\dd z) . \label{eq:Re.chi}
\end{align}
It follows from \eqref{eq:assumption}, \eqref{eq:u.bound2},  \eqref{eq:Re.psi} and \eqref{eq:Re.chi}, that there exists a constant $c_2 < \infty$ such that
\begin{align}
\sup_{q_r \in \Rb} \Big| \Re \psi(u( q \beta , q \tfrac{1}{2} \beta(1-\beta) ),0) \Big|
	&<	c_2 , &
\sup_{q_r \in \Rb} \Big| \Re \chi(q) \Big|
	&<	c_2	. \label{eq:chi.bound}
\end{align}
Finally, from \eqref{eq:phih.assumption}, \eqref{eq:u.bound2} and \eqref{eq:chi.bound} we conclude that inequality \eqref{eq:bound-this} holds, justifying the use of Fubini's Theorem.

%
%

\bibliographystyle{chicago}
\bibliography{Bibtex-Master-3.05rrvd}	

%
%


\begin{figure}
\centering
\begin{tabular}{c|c}
Effect of jump size & Effect of jump intensity \\
\includegraphics[width=0.45\textwidth]{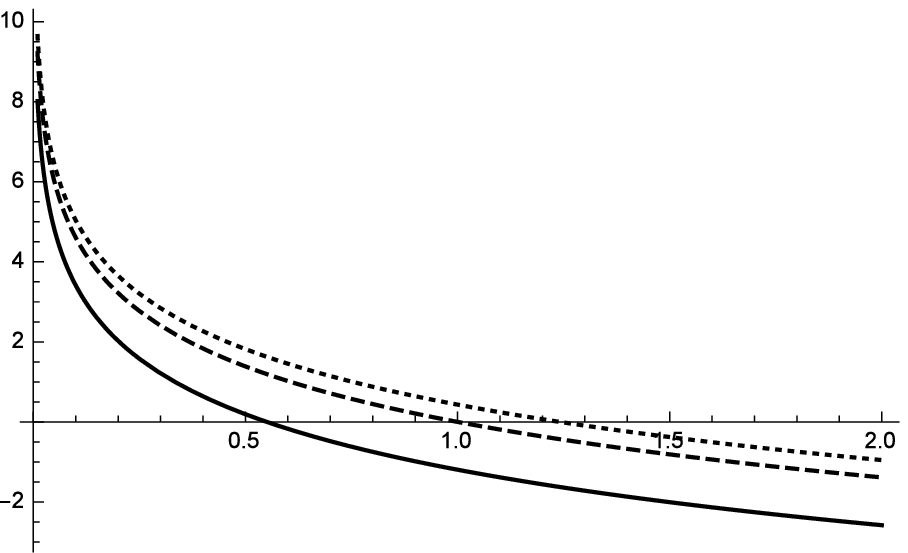} &
\includegraphics[width=0.45\textwidth]{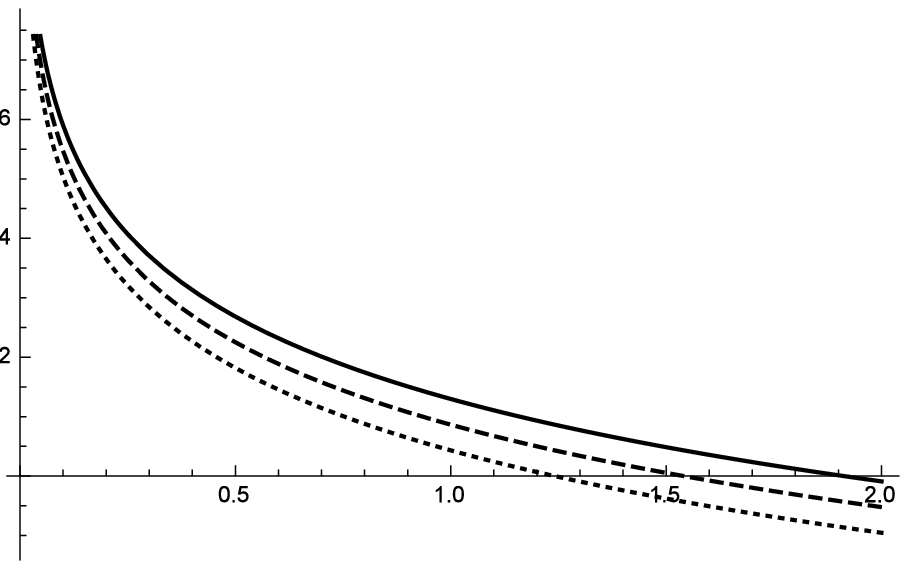} \\
$u=u_+$ & $u=u_+$ \\
\includegraphics[width=0.45\textwidth]{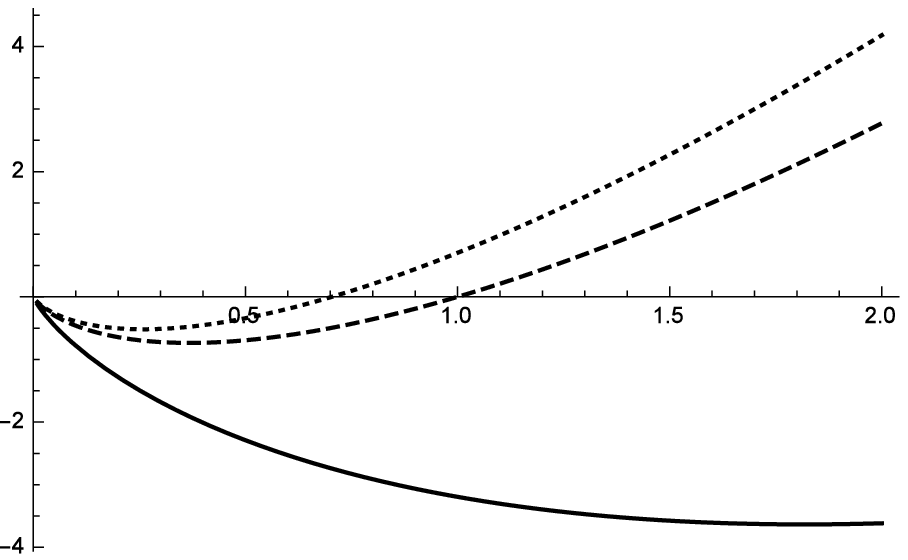} &
\includegraphics[width=0.45\textwidth]{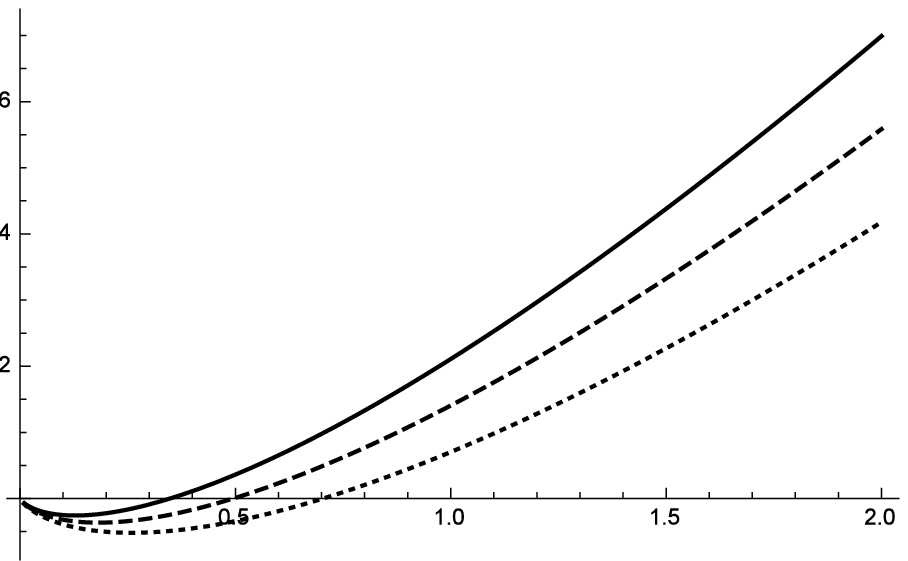} \\
$u=u_-$ & $u=u_-$ 
\end{tabular}
\caption{
We consider a Dirac L\'evy measure $\nu(\dd z) = \lam \del_m(z)\dd z$,
a variance swap payoff $[X]_T$ and plot the function $g(\log S_T;0,0)$ that prices the variance swap as a function of $S_T$.
\emph{Left}:
We examine the effect of the jump size $m$ when $g$ is computed using both $u=u_+$ and $u=u_-$.
The jump intensity is fixed at $\lam = 1.0$ and we vary $m = \{-2.0,0.0,2.0 \}$ corresponding to the dotted, dashed and solid lines, respectively.  Note that negative jumps (dotted line, $m=-2.0$) raises the value of $g$ at all points relative to no jumps (dashed line, $m=0.0$), whereas positive jumps (solid line, $m=2.0$) lowers the value of $g$ relative to no jumps. 
\emph{Right}:
We examine the effect of the jump intensity $\lam$ when $g$ is computed using both $u=u_+$ and $u=u_-$.
The jump size is fixed at $m = -2.0$ and we vary $\lam = \{1.0,2.0,3.0 \}$ corresponding to the dotted, dashed and solid lines, respectively.
As the jump intensity increases, so does the value of $g$ at all points.  Had jumps been upward, we would have seen $g$ decreasing as the jump intensity increased.
In all four plots the time to maturity is fixed at $T = 0.25$ years.
}
\label{fig:VarSwap}
\end{figure}

\begin{figure}
\centering
\begin{tabular}{c|c}
Effect of jump size & Effect of jump intensity \\
\includegraphics[width=0.45\textwidth]{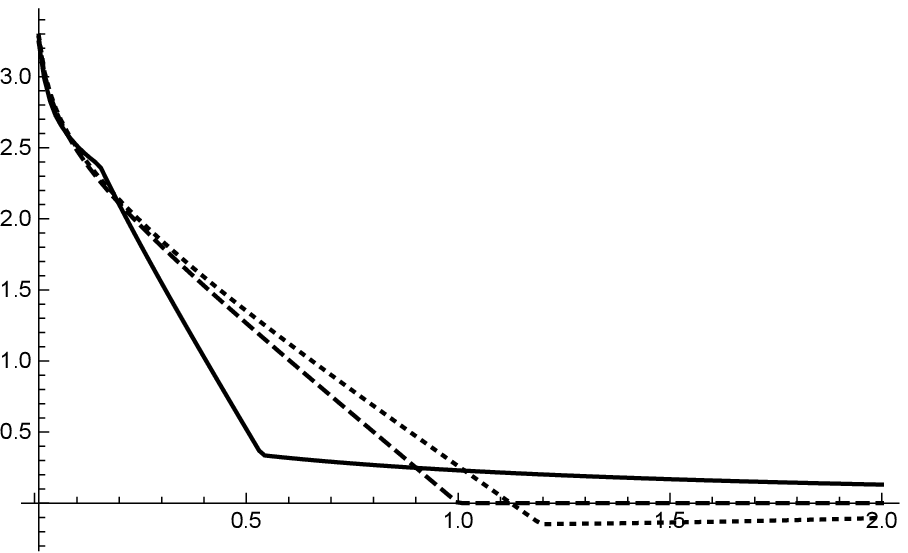} &
\includegraphics[width=0.45\textwidth]{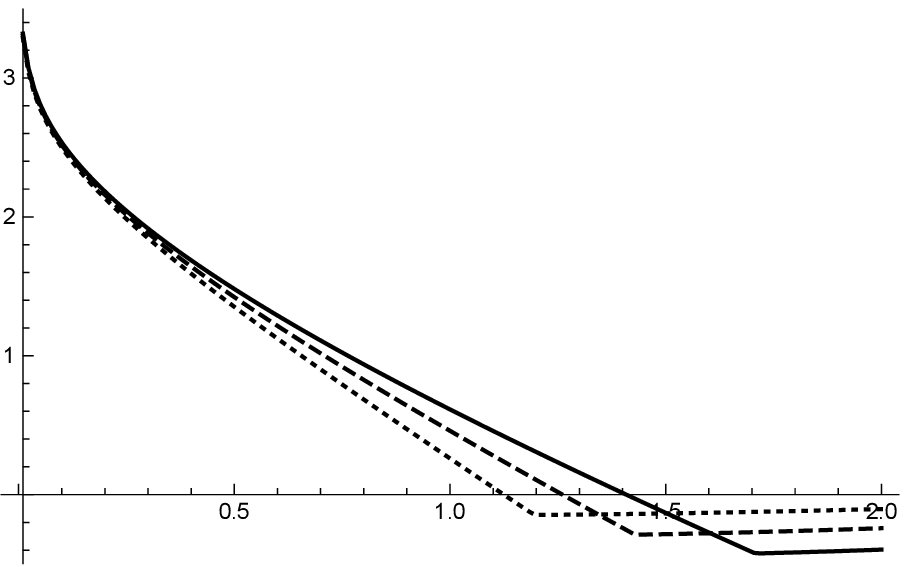} \\
$u=u_+$ & $u=u_+$ \\
\includegraphics[width=0.45\textwidth]{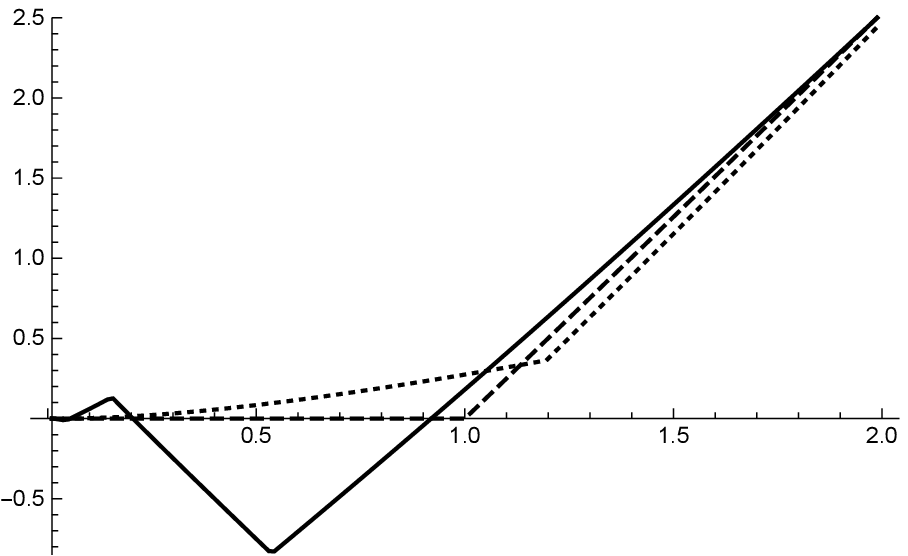} &
\includegraphics[width=0.45\textwidth]{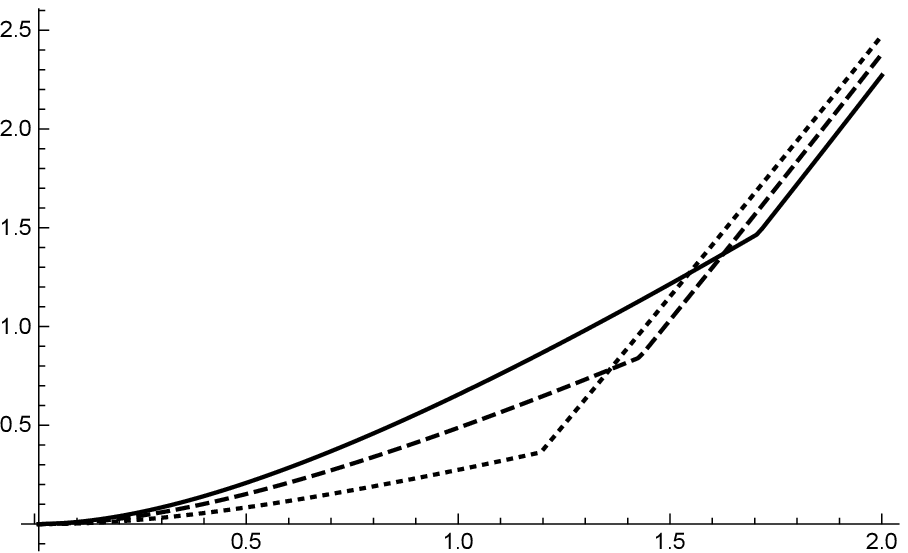} \\
$u=u_-$ & $u=u_-$ 
\end{tabular}
\caption{
We consider a Dirac L\'evy measure $\nu(\dd z) = \lam \del_m(z)\dd z$,
a volatility swap payoff $\sqrt{[X]_T}$ and plot the function $g(\log S_T)$ the prices the volatility swap as a function of $S_T$.
\emph{Left}:
We examine the effect of the jump size $m$ both for $u=u_+$ and for $u_-$.
The jump intensity is fixed at $\lam = 1.0$ and we vary $m = \{-1.25,0.00,1.25 \}$ corresponding to the dotted, dashed and solid lines, respectively.
\emph{Right}:
We examine the effect of the jump intensity $\lam$ both for $u=u_+$ and for $u=u_-$.
The jump size is fixed at $m = -1.25$ and vary $\lam = \{1.00,2.00,3.00 \}$ corresponding to the dotted, dashed and solid lines, respectively.
In all four plots the time to maturity is fixed at $T = 0.25$ years.
}
\label{fig:VolSwap}
\end{figure}

\begin{figure}
\centering
\begin{tabular}{cc}
\includegraphics[width=0.45\textwidth]{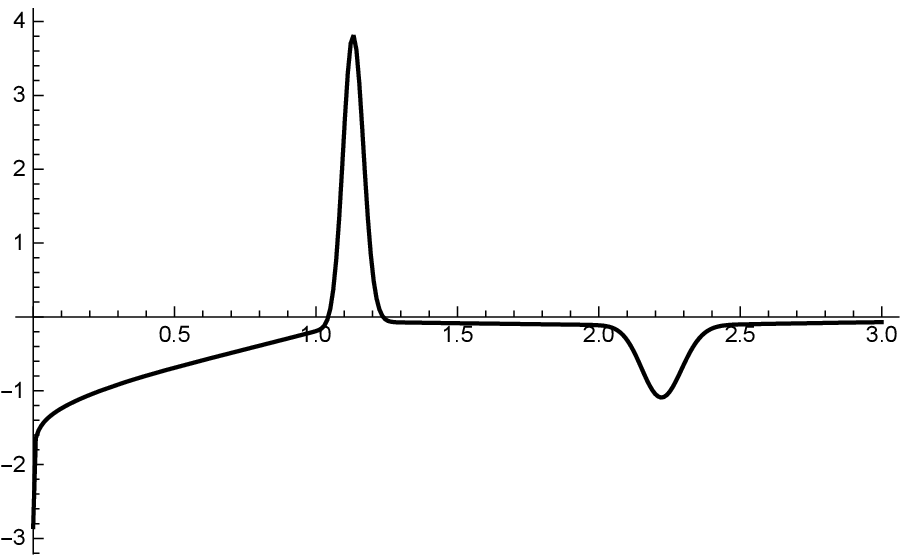} &
\includegraphics[width=0.45\textwidth]{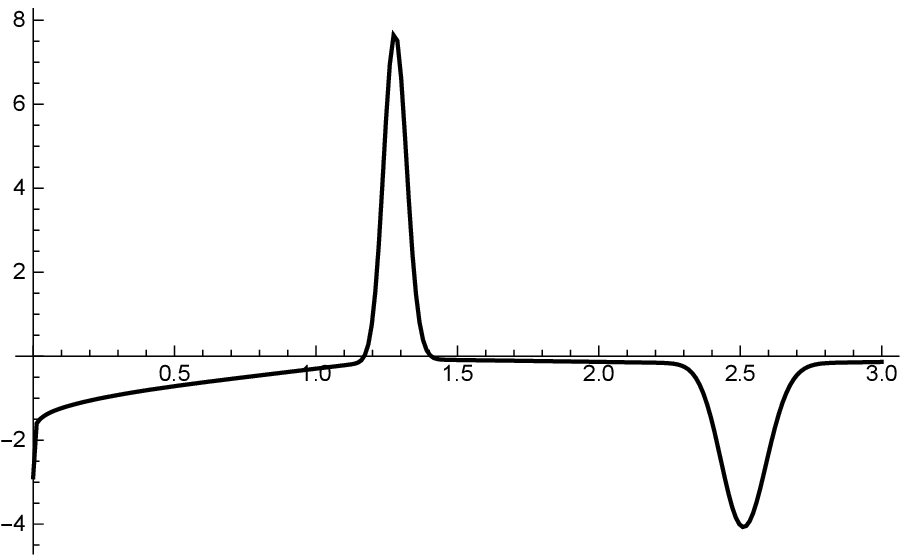} \\
$\lam = 1.0$, $m = -0.675$ & $\lam = 2.0$, $m = -0.675$ \\
\includegraphics[width=0.45\textwidth]{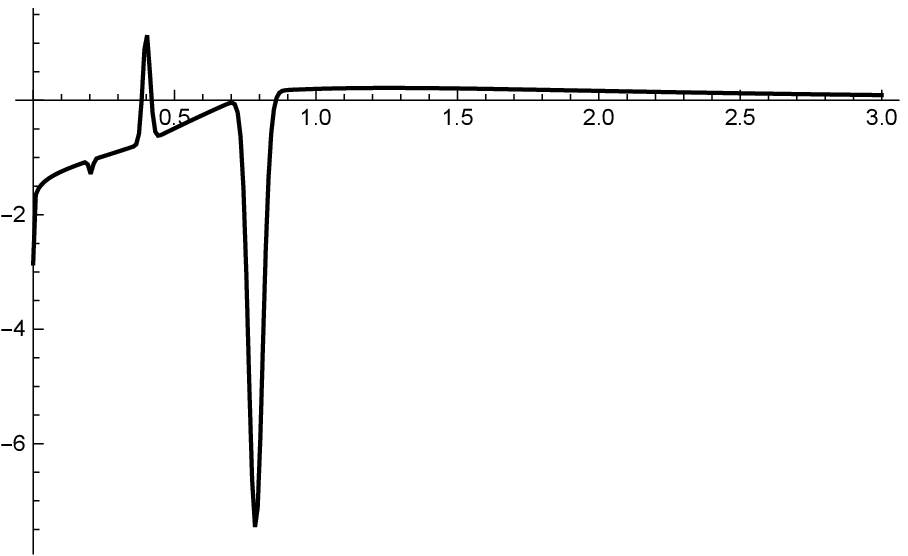} &
\includegraphics[width=0.45\textwidth]{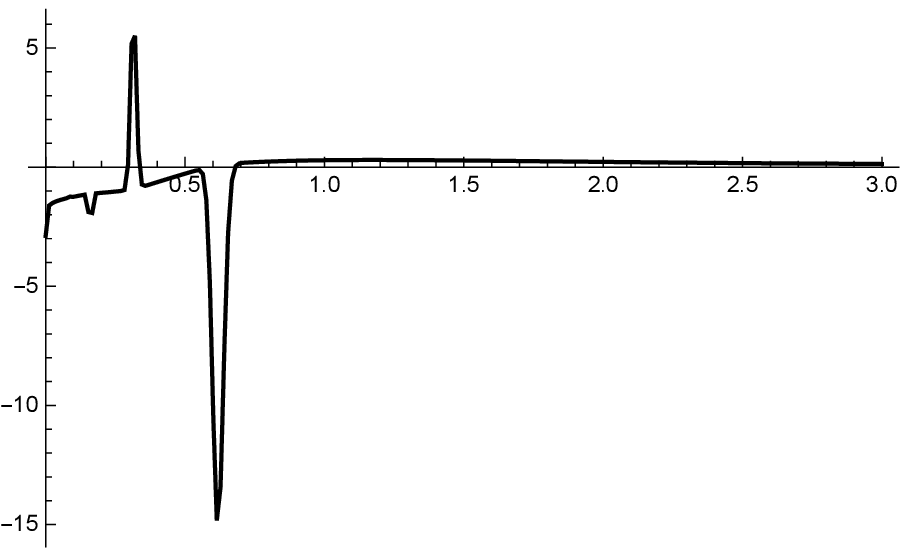} \\
$\lam = 1.0$, $m = 0.675$ & $\lam = 2.0$, $m = 0.675$ \\
\end{tabular}
\caption{We consider a Dirac L\'evy measure $\nu(\dd z) = \lam \del_m(z)\dd z$,
an approximate realized Sharpe ratio payoff $X_T/\sqrt{[X]_T+\eps}$ and plot the function $g(\log S_T)$  that prices this claim as a function of $S_T$.
In all four plots the time to maturity is fixed at $T=0.25$ years.  The parameter $\eps=0.001$ is fixed and we compute $g$ using $u=u_+$.
}
\label{fig:SharpeRatio}
\end{figure}

\begin{figure}
\centering
\begin{tabular}{cc}
\includegraphics[width=0.45\textwidth]{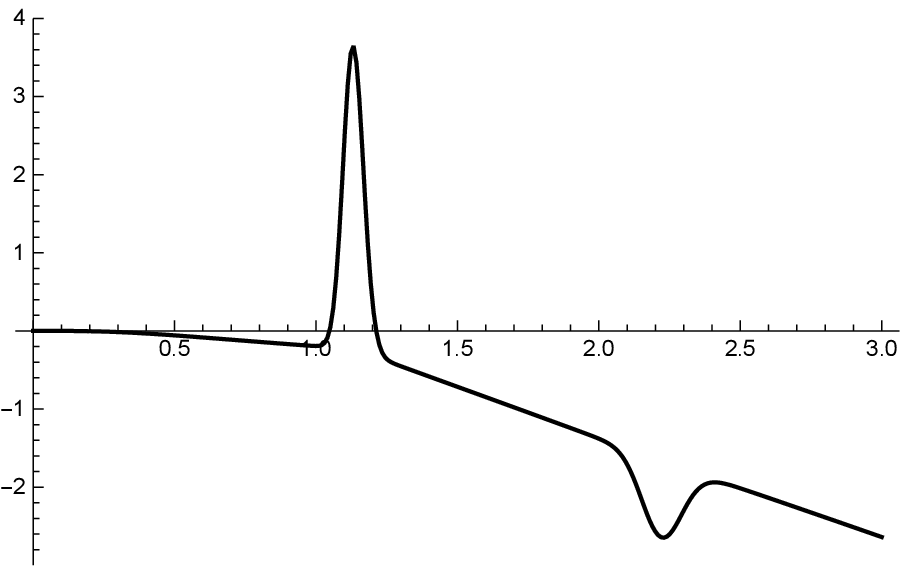} &
\includegraphics[width=0.45\textwidth]{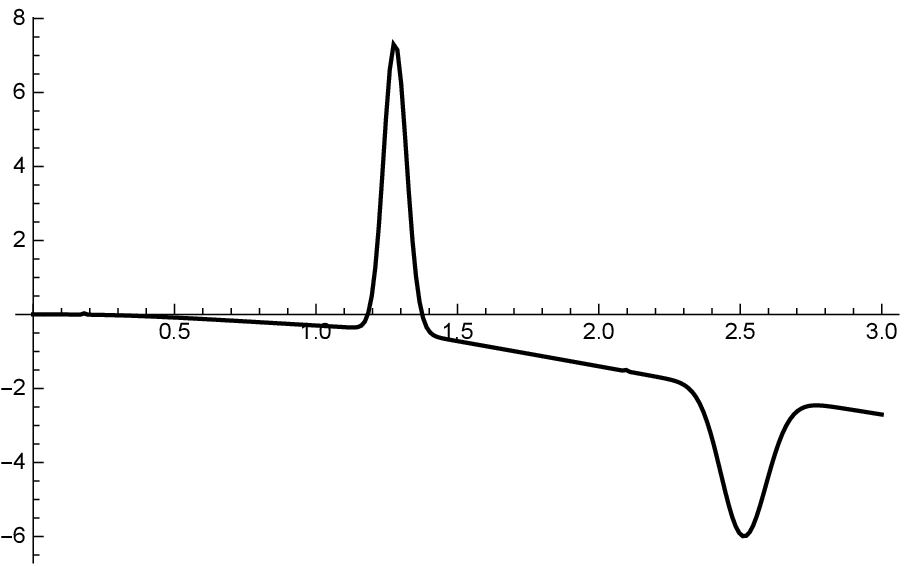} \\
$\lam = 1.0$, $m = -0.675$ & $\lam = 2.0$, $m = -0.675$ \\
\includegraphics[width=0.45\textwidth]{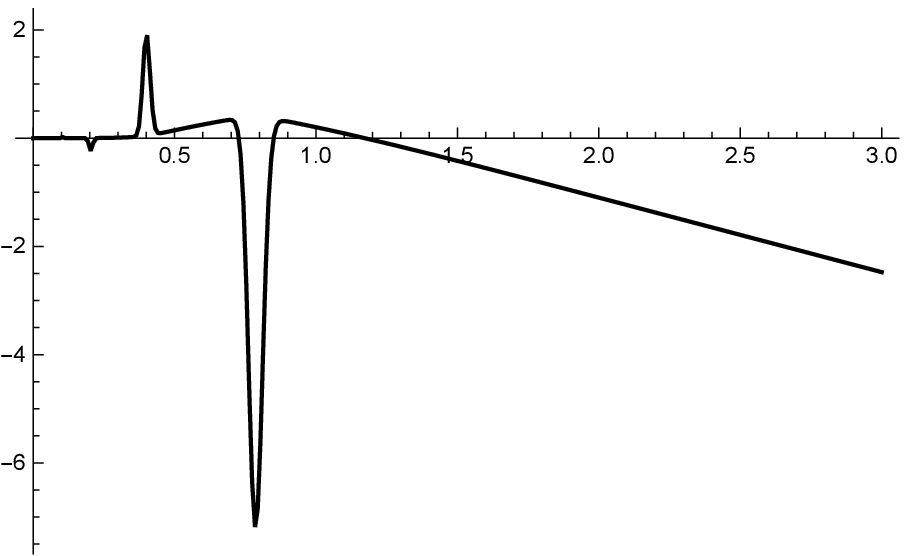} &
\includegraphics[width=0.45\textwidth]{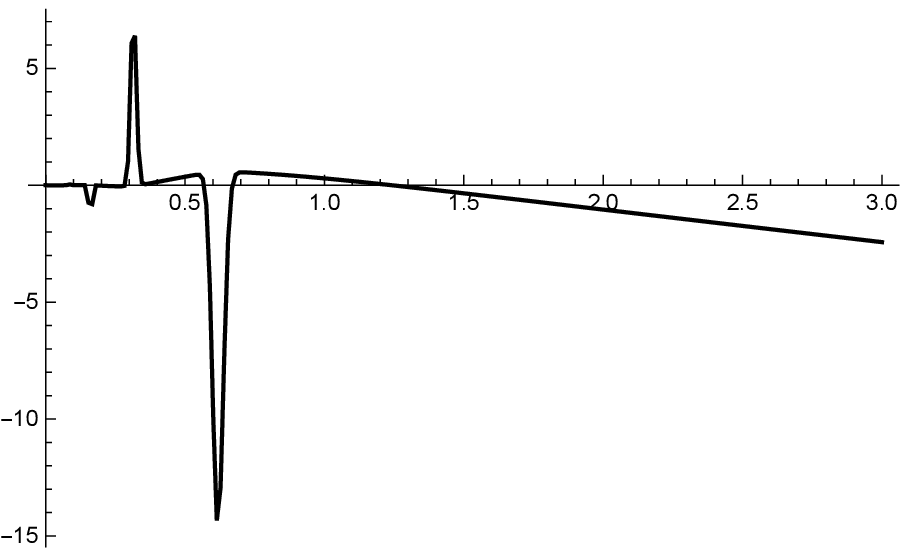} \\
$\lam = 1.0$, $m = 0.675$ & $\lam = 2.0$, $m = 0.675$ \\
\end{tabular}
\caption{We consider a Dirac L\'evy measure $\nu(\dd z) = \lam \del_m(z)\dd z$,
an approximate realized Sharpe ratio payoff $X_T/\sqrt{[X]_T+\eps}$ and plot the function $g(\log S_T)$  that prices this claim as a function of $S_T$.
In all four plots the time to maturity is fixed at $T=0.25$ years.  The parameter $\eps=0.001$ is fixed and we compute $g$ using $u=u_-$.
}
\label{fig:SharpeRatio2}
\end{figure}

\begin{figure}
\centering
\begin{tabular}{cc}
\includegraphics[width=0.45\textwidth]{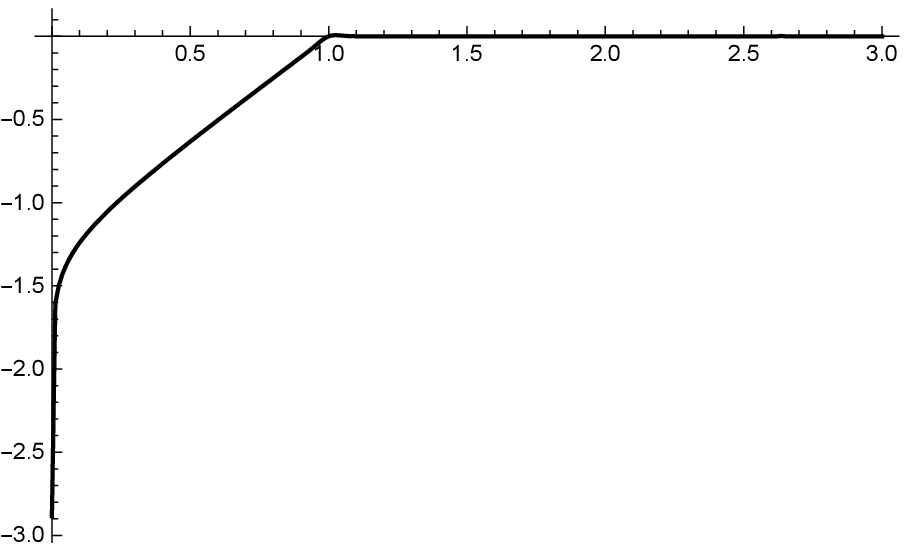} &
\includegraphics[width=0.45\textwidth]{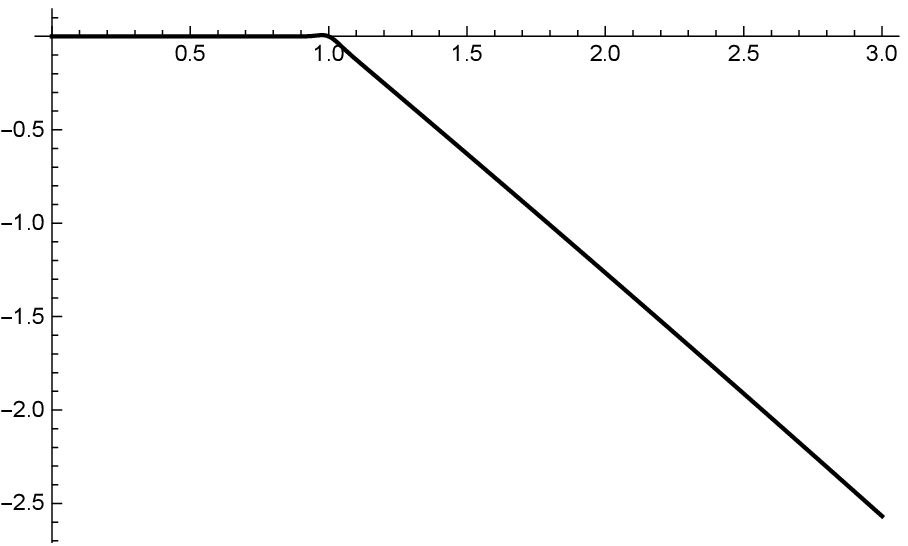} \\
$u_+$ & $u_-$ \\
\end{tabular}
\caption{We consider a L\'evy measure that is identically zero $\nu \equiv 0$,
an approximate realized Sharpe ratio payoff $X_T/\sqrt{[X]_T+\eps}$ and plot the function $g(\log S_T)$  that prices this claim as a function of $S_T$.
In both plots the time to maturity is fixed at $T=0.25$ years.  The parameter $\eps=0.001$ is fixed.
}
\label{fig:SharpeRatio3}
\end{figure}

\begin{figure}
\centering
\begin{tabular}{c|c}
$\beta > 0$ & $\beta<0$ \\
\includegraphics[width=0.45\textwidth]{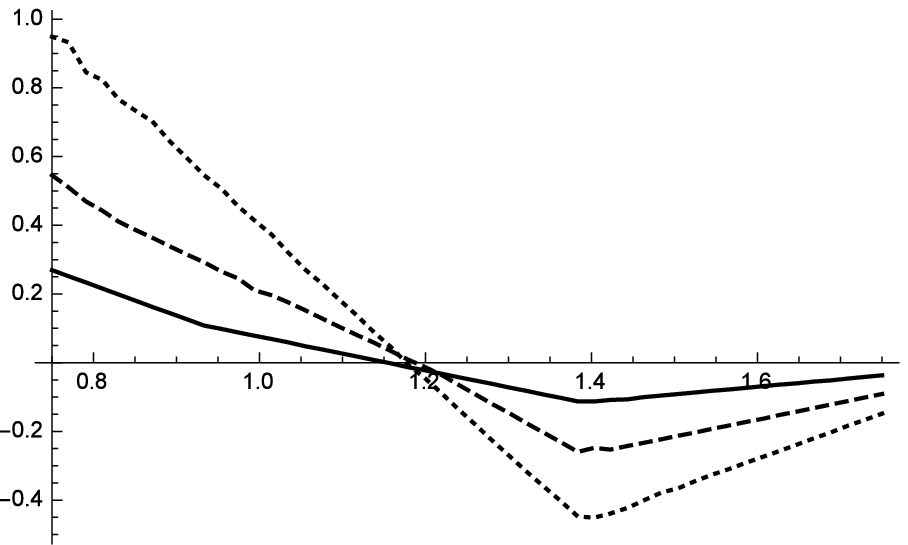} &
\includegraphics[width=0.45\textwidth]{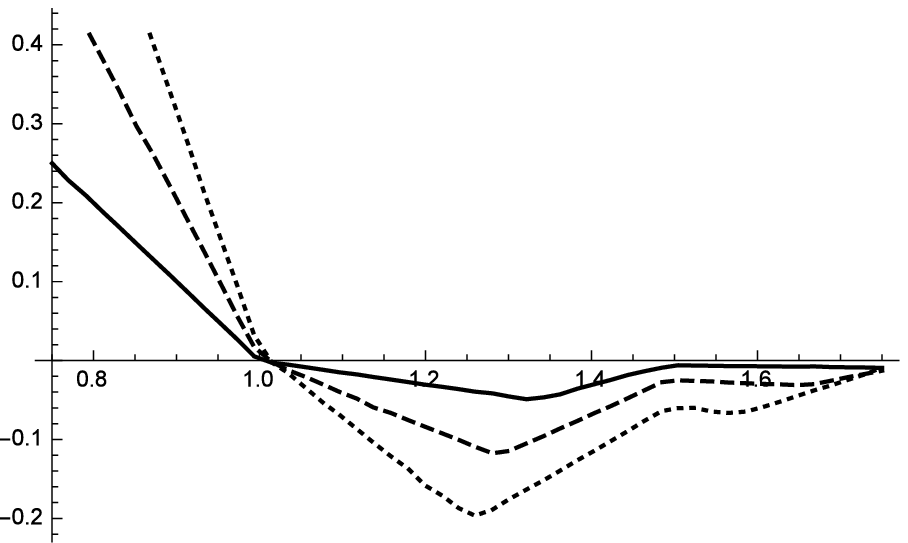} \\
$u=u_+$ & $u=u_+$ \\
\includegraphics[width=0.45\textwidth]{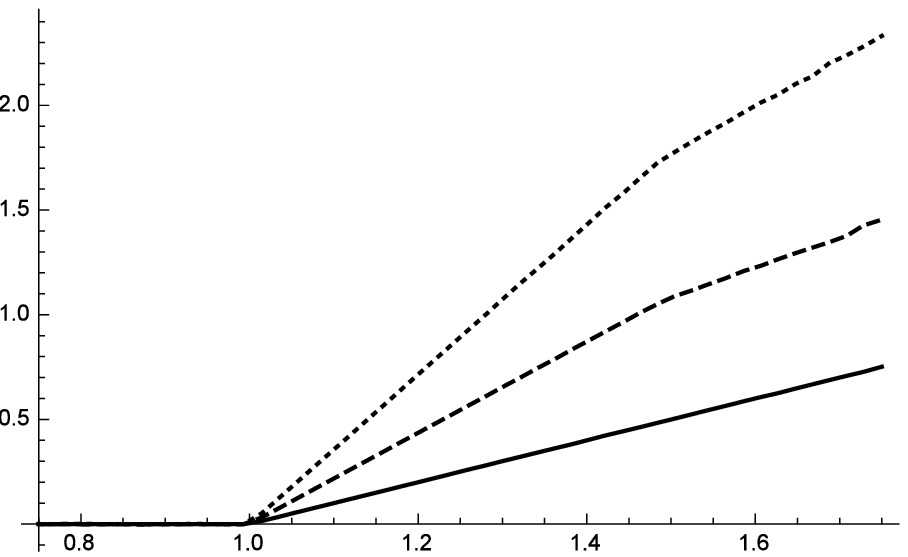} &
\includegraphics[width=0.45\textwidth]{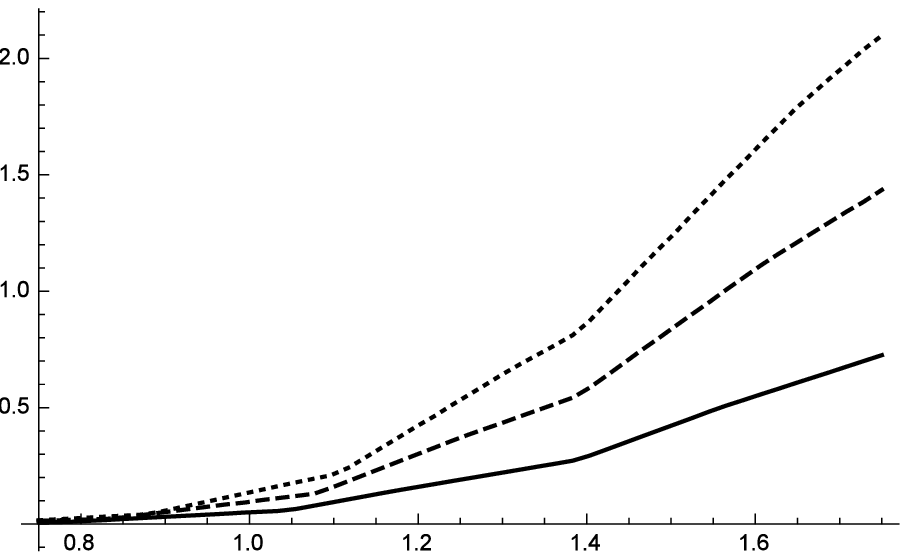} \\
$u=u_-$ & $u=u_-$ 
\end{tabular}
\caption{
Consider a call option written on an LETF $L=\ee^Y$.  In Theorem \ref{thm:LETF}, we provide an expression \eqref{eq:g.LETF} for a function $g$ that satisfies $\Eb_t (L_T - \ee^k)^+ = \Eb_t g \log S_T; X_t,Y_t )$, where $S=\ee^X$ is the underlying ETF.
In the plots above, we consider a Dirac L\'evy measure $\nu(\dd z) = \lam \del_m(z)\dd z$,
and plot $g_0(\log S_T;X_t,Y_t)$ as a function of $S_T$.
\emph{Left}:
For both $u=u_+$ and $u=u_-$, we consider positive leverage ratios $\beta = \{1,2,3\}$, corresponding to the solid, dashed, and dotted lines, respectively.
\emph{Right}:
For both $u=u_+$ and $u=u_-$, we consider negative leverage ratios $\beta = \{-1,-2,-3\}$, corresponding to the solid, dashed, and dotted lines, respectively.
In all four plots the following parameters are fixed $T = 0.25$ years, $X_0 = 0$, $Y_0=0$, $m=-0.4$, $\lam = 2.0$ and $k = 0$.
With $m$ as given, inequality \eqref{eq:assumption2} is satisfied for all six values of $\beta$.
Note that when $\beta = 1$, we have $L = S$.  Not surprisingly, when $u = u_-$, it appears that $g(\log S_T;X_t,Y_t) = (S_T - \ee^k)^+$ (solid line in the lower left plot).
}
\label{fig:LETF}
\end{figure}

\end{document}